\documentclass[pre,aps,twocolumn,,amsmath,amssymb,showpacs,showkeys]{revtex4-1}
\usepackage{graphicx}
\usepackage{amsmath,amssymb}
\usepackage{amsfonts}
\usepackage{mathrsfs}
\usepackage{epsfig}
\usepackage{dcolumn}
\usepackage{enumerate}
\usepackage{amsthm}
\usepackage{color}
\usepackage{xfrac}
\begin{document}

\newcommand{\newc}{\newcommand}
\newtheorem{proposition}{Proposition}
\newc{\beq}{\begin{equation}}
\newc{\eeq}{\end{equation}}
\newc{\kt}{\rangle}
\newc{\br}{\langle}
\newc{\beqa}{\begin{eqnarray}}
\newc{\eeqa}{\end{eqnarray}}
\newc{\pr}{\prime}
\newc{\longra}{\longrightarrow}
\newc{\ot}{\otimes}
\newc{\rarrow}{\rightarrow}
\newc{\h}{\hat}
\newc{\bom}{\boldmath}
\newc{\btd}{\bigtriangledown}
\newc{\al}{\alpha}
\newc{\be}{\beta}
\newc{\ld}{\lambda}
\newc{\sg}{\sigma}
\newc{\p}{\psi}
\newc{\eps}{\epsilon}
\newc{\om}{\omega}
\newc{\mb}{\mbox}
\newc{\tm}{\times}
\newc{\hu}{\hat{u}}
\newc{\hv}{\hat{v}}
\newc{\hk}{\hat{K}}
\newc{\ra}{\rightarrow}
\newc{\non}{\nonumber}
\newc{\ul}{\underline}
\newc{\hs}{\hspace}
\newc{\longla}{\longleftarrow}
\newc{\ts}{\textstyle}
\newc{\f}{\frac}
\newc{\df}{\dfrac}
\newc{\ovl}{\overline}
\newc{\bc}{\begin{center}}
\newc{\ec}{\end{center}}
\newc{\dg}{\dagger}
\newc{\prh}{\mbox{PR}_H}
\newc{\prq}{\mbox{PR}_q}
\newc{\tr}{\mbox{tr}}
\newc{\pd}{\partial}
\newc{\qv}{\vec{q}}
\newc{\pv}{\vec{p}}
\newc{\dqv}{\delta\vec{q}}
\newc{\dpv}{\delta\vec{p}}
\newc{\mbq}{\mathbf{q}}
\newc{\mbqp}{\mathbf{q'}}
\newc{\mbpp}{\mathbf{p'}}
\newc{\mbp}{\mathbf{p}}
\newc{\mbn}{\mathbf{\nabla}}
\newc{\dmbq}{\delta \mbq}
\newc{\dmbp}{\delta \mbp}
\newc{\T}{\mathsf{T}}
\newc{\J}{\mathsf{J}}
\newc{\sfL}{\mathsf{L}}
\newc{\C}{\mathsf{C}}
\newc{\B}{\mathsf{M}}
\newc{\V}{\mathsf{V}}

\title{Simple permutation-based measure of quantum correlations and maximally-3-tangled states}

\author{Udaysinh T. Bhosale}
\email{udaybhosale0786@gmail.com}
\affiliation{Indian Institute of Science Education and Research, Dr. Homi Bhabha Road, Pashan, Pune, 411 008, India.}

\author{Arul Lakshminarayan}
\email{arul@physics.iitm.ac.in}
\affiliation{Max-Planck-Institut f\"ur Physik komplexer Systeme, N\"othnitzer Stra\ss{}e 38, 01187 Dresden, Germany}
\affiliation{Department of Physics, Indian Institute of Technology Madras, Chennai, India~600036}

\date{\today}

\begin{abstract}

Quantities invariant under local unitary transformations are of natural interest in the study of entanglement. 
This paper deduces and studies a particularly simple quantity that is constructed from a combination of two 
standard permutations of the density matrix, namely realignment and partial transpose. This bipartite quantity, 
denoted here as $R_{12}$, vanishes on large classes of separable states including classical-quantum correlated states,
while being maximum for only maximally entangled states. It is shown to be naturally related to the 3-tangle in
three qubit states via their two-qubit reduced density matrices. Upper and lower bounds on concurrence and negativity of 
two-qubit density matrices for all ranks are given in terms of $R_{12}$. Ansatz states satisfying these bounds are given 
and verified using various numerical methods. In rank-2 case it is shown that the states satisfying the lower bound on $R_{12}$ {\it vs} concurrence 
define a class of three qubit states that maximizes the tripartite entanglement (the 3-tangle) given an amount 
of entanglement between a pair of them. The measure $R_{12}$ is conjectured, via numerical sampling, to be 
always larger than the concurrence and negativity. In particular this is shown to be true for the physically 
interesting case of $X$ states.

\end{abstract}

\maketitle

\section{Introduction}
\label{sec:Introduction}

Quantum entanglement, the nonlocal and unique feature of quantum mechanics, has been extensively investigated 
in the recent past, and forms an important part of quantum information theory \cite{Horodeckirpm}. 
In particular, shared bipartite entanglement is a crucial resource for many quantum information 
tasks such as teleportation \cite{Teleport}, quantum cryptography \cite{BennettProceedings}, entanglement swapping 
\cite{Horne93,Vedral98}, remote state preparation \cite{CharlesBennett01}, dense coding \cite{Superdense},
channel discrimination \cite{Pianiwatrous} and quantum repeaters \cite{Briegel99}.
Given a quantum state of many particles, say $\rho$, quantifying its entanglement content is naturally important. 
{This question has been settled in favor of the von Neumann entropy of the reduced density matrices in the case of pure bipartite states \cite{Bennett96}, as it quantifies the entanglement that can be concentrated using local operations alone.}  In the case of two-qubit states, mixed or pure,  ``concurrence"   is used, as it was shown 
to be a monotonic function of the entanglement of formation \cite{Wootters01,wootters98,Wooters}.

The partial transpose (PT) introduced by Peres \cite{peres96} is a powerful and simple tool to detect entanglement in mixed bipartite states.  However, while positivity under partial transpose is necessary for separability, it fails to detect a class of entangled states the so-called bound entangled states \cite{mhorodeckibound}. 
Nevertheless the logarithmic negativity \cite{logneg} measure based on the partial transpose is a useful measure of entanglement in mixed states. More general measures of
quantum correlation, such as the quantum discord \cite{ZurekQuantumDiscord2001,ZurekQuantumDiscord2000} have been extensively studied as well. These measures of correlations can be nonzero for states that have no entanglement content. 

{The purpose of the present work is to focus on a measure that uses a simple permutation of the density matrix. Given the multitude of entanglement and correlation measures, this work seeks to highlight some of the unique properties of this quantity such as its vanishing on classically correlated states and its natural relationship  with concurrence and the three-tangle. It is seen as a natural quantity when considering the problem of maximizing 3-tangle of three qubit states given a fixed entanglement between two of them. This gives rise to an interesting set of states that we refer to as maximally 3-tangled states. Thus it is likely that this measure, which can be easily calculated for any state, has physical content that deserves further exploration.} 

{Any measure of entanglement can not increase under local operations and classical communications 
(LOCC) and it should be constant and minimal on all separable states \cite{Plenio07,Horodeckirpm}. This also implies that the measure of entanglement must be invariant under local unitary 
(LU) transformations. 
}
The spectra of the density matrix itself and the various reduced density matrices got by tracing out subsystems are such LU invariants. They could be invariant under non-local operations and therefore their interpretation in terms of entanglement is generally not tenable.
However, interestingly, if the single subsystem reduced density matrices of a multipartite pure state are considered, collection of their eigenvectors form convex polytopes that characterize distinct entanglement classes \cite{AdamSawicki12,MichaelWalter2013,AdamSawicki14}.  
Here, the notion of entanglement class is a broader class than  LU, and includes measurements and classical 
communications, that are included in the operation known as stochastic local operations and classical 
communications (SLOCC). It is clear that states that cannot be converted to each other by LU are also not SLOCC equivalent, but the converse is not true. States that can be converted to each other by 
SLOCC are {from} 
an entanglement class. In the case of three qubit pure states there are 2 different entanglement classes, known as the $W$ and the GHZ, while for 4 qubits there are nine \cite{FVerstraete2002}.

Given a bipartite system $1$ and $2$ having a product orthonormal basis
\{$|i\rangle |\alpha\rangle$\} and density matrix $\rho_{12}$, the PT with respect to the second subsystem, denoted as $\rho_{12}^{T_2}$, is given by the matrix elements:
\begin{equation}
(\rho_{12}^{T_2})_{i\alpha;j\beta} = (\rho_{12})_{i\beta;j\alpha}\;;\;\;\; (\rho_{12})_{i\alpha;j\beta}= 
\langle i|\langle \alpha|\rho_{12}|j\rangle|\beta \rangle.
\end{equation}
Peres's partial transpose criterion states that if $\rho_{12}^{T_2}$ 
is negative then the state $\rho_{12}$ is entangled. 
The other operation of interest to the present work is realignment \cite{Chen03,Oliver04}.
The corresponding operation on the density matrix $\rho_{12}$, denoted as 
$\mathcal{R}(\rho_{12})$ is given by:
\beq
\langle i|\langle j|\left(\mathcal{R}(\rho_{12})\right)|\alpha\rangle|\beta \rangle= \langle i|\langle \alpha|\rho_{12}|j\rangle|\beta \rangle.
\label{Eq:RealDef}
\eeq
The realignment criterion is that if the state $\rho_{12}$ is separable then $\Vert \mathcal{R}(\rho_{12})\Vert_1 \le1$, where 
$\Vert M \Vert_1$ is the  trace norm equal to $\mbox{tr}\sqrt{M M^{\dagger}}$ \cite{Oliver03}. 
This condition is found to detect some bound entangled states, these being positive under PT and hence not being detected by the corresponding criterion \cite{Oliver04,mhorodeckibound}. 
Note that both the realignment and partial transpose are simple 
permutations of the elements of the density matrix. While the partial transpose retains the Hermiticity of the operator, 
realignment does not. In fact if the subsystems are of different dimensionalities, it results in a rectangular matrix.

Consider a system consisting of $M$ subsystems labeled by $i$, ($1 \le i \le M$) each in a Hilbert space of dimension $d_i$. If its joint state is $\rho$ it was shown in \cite{Udaysinh13}, motivated by considerations in \cite{Williamson11}, that LU invariants could be constructed in the following way. Let $\{a\}=(i_1,i_2,\cdots, i_K)$ be an arbitrary {\it closed} path in the space of labels of the subsystems, that is $1 \le i_k \le M$ and $i_{K+1}=i_1$. Then the eigenvalues of 
\beq
{\mathcal P}\left(\{a\}\right)={\mathcal R}\left(\rho_{i_1 i_K}^{T_{i_{K}}}\right)
\cdots {\mathcal R}\left(\rho_{i_2 i_1}^{T_{i_1}}\right)
\eeq
are (in general complex) LU invariants. Here, $\rho_{i_k i_m}$ is a bipartite state obtained by tracing out all other subsystems except $i_k$ and $i_m$. It was also shown in \cite{Udaysinh13} that the characteristic polynomial of ${\mathcal P}$ is real and hence these real coefficients are also LU invariants.
Note that the ``link transformation'' \cite{Udaysinh13,Williamson11} is a combination of both PT and realignment, executed in that order.

In the same work \cite{Udaysinh13}, it was shown that for the case of bipartite pure state of two qubits
the quantity $\det[\mathcal{R}(\rho_{12}^{T_2})\mathcal{R}(\rho_{21}^{T_1})]^{1/4}$ is equal to 
$\tau_2/4$, where $\tau_2$ is the two-tangle \cite{Coffman}, {\it i.e.} square of the concurrence 
\cite{Wooters,wootters98,Wootters01}. This motivates the question of the meaning of the spectra of $\mathcal{R}(\rho_{12}^{T_2})$ as entanglement or correlation measures in an arbitrary bipartite system. Using this,  a new and central quantity of this paper is defined and investigated in the case of two-qubit density matrices of general rank.

The structure of the paper is as follows:
In Sec.~\ref{sec:SimpleLUI} the measure is defined and basic properties are studied, along with 
its evaluation for various classes of states. In Sec.~\ref{sec:ThreeQubitPureStatesChp5} two-qubit states of rank-2 or
equivalently three-qubit pure states are studied in detail. Bounds for the concurrence are given in terms of the defined measure. Various boundaries of the inequality are characterized and two classes of maximally 3-tangled states are discussed. The difference between the concurrence and the measure is shown to have direct connections with the tripartite measure of 3-tangle. In Sec.~\ref{sec:2qubithigherrank} results on two-qubit states of rank greater than two are presented. Rank-3 and rank-4 boundaries are also investigated and all of these are found to be Bell-diagonal states.
The measure is also evaluated for a special class of states namely the $X$ states and is shown to be larger than the concurrence. In Sec.~\ref{sec:ComR12N12} negativity is compared with $R_{12}$ and again various boundaries are investigated and their significance pointed out when we can.
For example in this case the MEMS I states lie on the boundary for rank-2 states. The Werner states and pure states form common outer boundaries in both the $R_{12}$-concurrence and $R_{12}$-negativity comparisons.

\section{A simple LU invariant from realignment and partial transpose}
\label{sec:SimpleLUI}
 
As discussed above, products of certain permutations  of bipartite density matrices along a path in the space of labels are capable of generating LU invariants. This paper is mainly devoted to exploring the simplest of these, namely when the path simply connects two subsystems: say $1 \rarrow 2 \rarrow 1$. The operator in this case is 
${\mathcal P}(12)={\mathcal R}(\rho_{1 2}^{T_{2}}){\mathcal R}(\rho_{2 1}^{T_{1}})= {\mathcal R}(\rho_{1 2}^{T_{2}}){\mathcal R}(\rho_{1 2}^{T_{2}})^{\dagger}$, and hence is positive. While all the eigenvalues or the coefficients of the characteristic polynomial of ${\mathcal P}(12)$ maybe considered, in this paper the quantity 
\beq
R_{12}=d \, \left(\det \left[{\mathcal P}(12)\right]\right)^{1/2d^2}=  d \, \left(\left|\det[\mathcal{R}(\rho_{12}^{T_2})] \right|\right)^{1/d^2},
\label{eq:R12}
\eeq
is studied. In particular the case $d_1=d_2=d$, is considered so that the array $\mathcal{R}(\rho_{12}^{T_2})$ is square. Else, straightforward generalized forms need to be used. The somewhat strange powers is to make contact with the well-known entanglement measure of concurrence when $d=2$, a case that we will almost exclusively consider. From the definition it is obvious that $R_{12}/d$ is the geometric mean of the singular values of $\mathcal{R}(\rho_{12}^{T_2})$.

{
This quantity contains the entanglement along with other correlations that may come from multipartite 
entanglement of purifications of $\rho_{12}$. 
Various evidences for this fact will be presented in the subsequent 
parts of the paper. To be precise it is shown that for various class of states
$R_{12}$ exceeds or equals well known measures of entanglement which shows that the term $\rho_{12}$ captures
other correlations along with entanglement.
It is interesting to note that in the case of two qubits, $R_{12}$ is equal to the
volume of a steering ellipsoid \cite{TerryRudolph2014,AntonyMilne2014,AntonyMilne2014BB} also 
called ``obesity" which arises in the quantum steering ellipsoid
formalism of two-qubit states.
The quantum steering ellipsoid of a two-qubit state is defined as the set of Bloch vectors that Bob can collapse
Alice’s qubit to, considering all possible measurements on his qubit.
This formalism has provided a faithful and intuitive representation of two-qubit states.
This observation gives an operational meaning to the term $R_{12}$ that needs further exploration.

}
The letter $R$ is used to signify this quantity and maybe 
considered as some sort of ``rapprochement" between the two subsystems. As the partial transpose followed by 
realignment is repeatedly done in the following, it is useful to show their combined operation explicitly in the 
case of a two qubit state:
\beq
\begin{split}
\rho&=\left( \begin{array}{cccc} a_{11}&a_{12}&a_{13}&a_{14}\\a_{12}^*&a_{22}&a_{23}&a_{24}\\a_{13}^*&a_{23}^*&a_{33}&a_{34}\\a_{14}^*&a_{24}^*&a_{34}^*&a_{44}\end{array} \right) \mapsto \\ {\mathcal R}(\rho^{T_2}) &=
\left( \begin{array}{cccc} a_{11}&a_{12}^*&a_{12}&a_{22}\\a_{13}&a_{23}&a_{14}&a_{24}\\a_{13}^*&a_{14}^*&a_{23}^*&a_{24}^*\\a_{33}&a_{34}^*&a_{34}&a_{44}\end{array} \right). 
\end{split}
\eeq

{
For product states of the form 
$\rho_{12}=\rho_1 \otimes \rho_2$, which is the special case of separable states,  ${\mathcal R}(\rho_{12}^{T_2})$ is a rank-1 projector.
 The crucial observation is that for such product states 
\beq \label{eq:R12ProductState}
{\mathcal R}(\rho_{12}^{T_2})=\rho_1^R  (\rho_2^R)^{\dagger}, 
\eeq 
where the notation $\rho_{1,2}^R$ denotes the density matrix reshaped into a column vector of dimension $d_{1,2}^2$
and $(\rho_{1,2}^R)^{\dagger}$ denotes the Hermitian conjugate of $\rho_{1,2}^R$. 
The reshaping is done by stacking the rows in a column. Therefore, ${\mathcal R}(\rho_{12}^{T_2})$ is a matrix 
of dimension $d_1^2 \times d_2^2$. To understand this more clearly an example of two qubits will be considered 
explicitly. Let $\rho_1$ and $\rho_1$ denote the density matrices of qubits $1$ and $2$ respectively as follows:
\beq
\rho_1=\left( \begin{array}{cc}
	      a_{11}   & a_{12}\\
	      a_{12}^* & a_{22} \\
\end{array}\right)\;\;\;\mbox{and}\;\;\;
\rho_2=\left( \begin{array}{cc}
	      b_{11}   & b_{12}\\
	      b_{12}^* & b_{22} \\
\end{array}\right).
\eeq 
Then, $\rho_{1,2}^R$ are given as follows:
\beq
\rho_1^R=\left( \begin{array}{c}
	      a_{11}\\ 
	      a_{12}\\
	      a_{12}^* \\
	      a_{22} \\
\end{array}\right)
\;\;\;\mbox{and}\;\;\;
\rho_2^R=\left( \begin{array}{c}
	      b_{11}   \\
	      b_{12}\\
	      b_{12}^* \\
	      b_{22} \\
\end{array}\right)
\eeq 
Using Eq.~(\ref{eq:R12ProductState}) one obtains the following:
This implies that 
\beq
\rho_1^R  (\rho_2^R)^{\dagger}=
\left( \begin{array}{cccc}
	      a_{11}b_{11} & a_{11}b_{12}^*  & a_{11}b_{12} & a_{11}b_{22}\\ 
	      a_{12}b_{11} & a_{12}b_{12}^*  & a_{12}b_{12} & a_{12}b_{22}\\ 
	      a_{12}^*b_{11} & a_{12}^*b_{12}^*  & a_{12}^*b_{12} & a_{12}^*b_{22}\\ 
	      a_{22}b_{11} & a_{22}b_{12}^*  & a_{22}b_{12} & a_{2}b_{22}\\ 
\end{array}\right)
\eeq
and is easily seen to be ${\mathcal R}(\rho_{1}\otimes \rho_{2}^{T})$.
}

For two-qubit pure states it is easy to see that \cite{Udaysinh13} $R_{12}=C_{12}$, where $C_{12}$ is the 
concurrence \cite{Wooters,wootters98,Wootters01}. Thus, this motivates a more detailed study of this quantity in 
the case of higher rank two-qubit density matrices and its relationship to the concurrence. 
$R_{12}$ is a symmetric measure, {\it i.e.} $R_{12} = R_{21}$, when
the subsystems have equal dimensions. This follows from the definition, $R_{12}$ depends on the eigenvalues of  
${\mathcal R}(\rho_{1 2}^{T_{2}}){\mathcal R}(\rho_{21}^{T_{1}})$, while 
 $R_{21}$ depends on the eigenvalues of  
${\mathcal R}(\rho_{2 1}^{T_{1}}){\mathcal R}(\rho_{12}^{T_{2}})$. These two sets of eigenvalues can differ only
in the number of zero eigenvalues, which happens when the subsystem dimensions are different. 
Hence we may define $R_{12}$ to be such that subsystem labeled $1$ is not of larger dimension than that of subsystem $2$.

The following proposition is now proved:
\begin{proposition}{ $0 \le R_{12} \le 1.$} 
\label{prop:Rleq1}
\end{proposition}
\begin{proof}
Let the eigenvalues of ${\mathcal P}(12)={\mathcal R}(\rho_{1 2}^{T_{2}}){\mathcal R}(\rho_{1 2}^{T_{2}})^{\dagger}$, be $\mu_j$, $1 \le j \le d^2$. Then
\beq
R_{12}=d \left(\prod_{j=1}^{d^2} \sqrt{\mu_j}\right)^{1/d^2} \le \frac{1}{d} \sum_{j=1}^{d^2}\sqrt{\mu_j},
\eeq
which follows from the fact that geometric mean is no larger than the arithmetic mean. As ${\mathcal R}(\rho_{1 2}^{T_{2}})$ is only a permutation of the original density matrix, it follows that $\tr ({\mathcal P}(12))=\sum_{j=1}^{d^2} \mu_j=\tr (\rho_{12}^2)\le 1.$ An application of the Cauchy-Schwarz inequality: $\sum_i a_i b_i \le \sqrt{\sum_ i a_i^2} \sqrt{\sum_i b_i^2}$, with $a_i=\sqrt{\mu_i}$, $b_i=1$ gives  $\sum_{j=1}^{d^2}\sqrt{\mu_j} \le d$, which results in $R_{12} \le 1$ as required. The lower limit is evident.
\end{proof}

It is instructive to evaluate $R_{12}$ for well-known classes of states and therefore the following examples 
are considered.
\begin{enumerate}

\item {\it Bipartite pure states:}
Using Schmidt decomposition every bipartite pure state can always be written as:
$|\psi_{12}\rangle=\sum_{k=1}^{d}\sqrt{\lambda_k}|\phi_k\rangle|\psi_k\rangle$ where $d=\mbox{min}\{d_1,d_2\}$. 
It follows that ${\mathcal R}(|\psi_{12} \kt \br \psi_{12} |^{T_{2}})=$
\beq
 \sum_{k,j}\sqrt{\lambda_j \lambda_k} |\phi_j \kt |\phi_k\kt \br \psi_k|\br \psi_j|.
\eeq
The eigenvalues of ${\mathcal P}(12)$  are then  
$\lambda_k \lambda_j$ for $k$, $j=1,2,\ldots d$ which gives \[ R_{12}=d\,\left(\prod_k^d \lambda_k\right)^{1/d}. \]
For maximally entangled states $\lambda_k=1/d$  $\forall$  $k$, and it follows that $R_{12}=1$.
This also follows from the fact that the maximally entangled state is $\sum_{j=1}^d |jj \kt/\sqrt{d}$,
gives ${\mathcal R}(\rho_{12}^{T_2})=S_{12}/d$, where $S_{12}$ is the swap 
operator $S_{12}|i j\kt = |ji\kt$. 

\item {\it Bell-diagonal states:} These states, as the name suggests, are diagonal in the
Bell basis \cite{Wootersentform,Verstraete2001}: 
\[ \begin{split} \rho_{12}=&p_1|\phi^+\rangle\langle\phi^+|+p_2|\psi^+\rangle\langle\psi^+|\\&+p_3|\psi^+\rangle\langle\psi^+|+
p_1|\phi^-\rangle\langle\phi^-| \end{split} \] where $\sum_i^4 p_i=1$, and $|\psi^{\pm}\rangle=(1/\sqrt{2})(|01\rangle+|10\rangle)$ 
while $|\phi^{\pm}\rangle=(1/\sqrt{2})(|00\rangle+|11\rangle)$. This state is separable iff it's spectrum lies in 
$[0,1/2 ]$ \cite{RyszardHorodecki1996}, otherwise it is entangled. The entanglement calculated using the
concurrence is $C_{12}=\mbox{max}\{0,2p_{max}-1\}$ \cite{Wooters} where 
$p_{\mbox{max}}=\mbox{max}\{p_1,p_2,p_3,p_4\}$. It is a simple calculation to show that 
\begin{eqnarray*}
&&\;\;\;\;\;\;\;\;\;R_{12}=\\ \nonumber
&&\;\;\;\;\;\;\;\;\;\left|8 (p_2+p_3-\sfrac{1}{2}) (p_2+p_4-\sfrac{1}{2}) (p_3+p_4-\sfrac{1}{2})\right|^{1/4}\nonumber
\end{eqnarray*}
%
and is nonzero even when concurrence is zero.
In fact $R_{12}$ is zero for any pair of $p_i$ and $p_j$ ($i \neq j$) satisfying $p_i+p_j=1/2$. Bell-diagonal states
appear as boundaries in many of the following phase diagrams, and one special case of it, the Werner state is worth
singling out for further details. 

\item {\it Werner state:} A well-known mixture of the maximally entangled and mixed state is the two-qubit Werner 
state: $\rho_{12}=(1-p) I /4+p \, |\phi^{+}\rangle \langle \phi^{+}|$ \cite{WernerState1989}
 where $p$ ($0\leq p \leq 1$). This state is entangled iff $1/3 \leq p \leq 1$ and in that case the entanglement, 
as measured by the concurrence is $C_{12}=(3p-1)/2$. It is readily seen however that $R_{12}=p^{3/4}$, and hence 
is nonzero when the concurrence is zero, except in the extreme case of $p=0$, when there is a maximally mixed 
state. It is easy to see that $p^{3/4}-(3p-1)/2$ is monotonically decreasing in $[1/3,1]$ and hence attains the 
minimum value of $0$ at $p=1$. This implies that $R_{12}\geq C_{12}$ for all values of $p$, equality occurring 
only at the extreme cases of $p=0$ and $1$.

{
\item {\it Maximally Entangled Mixed States (MEMS):} 
These states
\cite{SatoshiIshizaka2000,Verstraete2001,WJMunro01,TzuChiehWei03,NicholasPeters04} are  
two-qubit states whose entanglement (concurrence) is maximized for a given value of mixedness,
measured using linear entropy. These states have been realized experimentally 
\cite{NicholasPeters04} using correlated photons from parametric down-conversion.
There are two classes of MEMS,  the rank-2 (MEMS~I) and rank-3 (MEMS~II) ones and are given as 
follows:
\begin{equation}
\begin{split}
\rho_{MEMS I} &=\left( \begin{array}{cccc}
    C_{12}/2 & 0 &  0 & C_{12}/2\\
	      0 & 1-C_{12} & 0 & 0\\
	      0 & 0 & 0 & 0 \\     
	      C_{12}/2 & 0 & 0 & C_{12}/2\\
\end{array}\right),\\
&  \mbox{where}\;\;\;\dfrac{2}{3} \leq C_{12} \leq 1,
\end{split}
\end{equation}
and
\begin{equation}
\begin{split}
\rho_{MEMS II}&=\left( \begin{array}{cccc}
    1/3 & 0 &  0 & C_{12}/2\\
	      0 & 1/3 & 0 & 0\\
	      0 & 0 & 0 & 0 \\     
	      C_{12}/2 & 0 & 0 & 1/3\\
\end{array}\right),\\
&  \mbox{where}\;\;\; 0 \leq C_{12} \leq \dfrac{2}{3} \;\;\; \mbox{respectively}.
\end{split}
\end{equation}
It readily follows from the definition in Eq.~(\ref{eq:R12}) that $R_{12}=C_{12}$ for MEMS I and 
$R_{12}=\sqrt{2C_{12}/3}$ for MEMS II for the respective ranges of $C_{12}$.
It is easy to see $\sqrt{2C_{12}/3}-C_{12}\geq 0$ in the range $0 \leq C_{12} \leq 2/3$, equality
occurring only at $C_{12}=0$ and $2/3$. This implies that $R_{12}\geq C_{12}$ for MEMS. 
We will see below that this inequality is of general validity.
}

\item {\it Separable states:}

Consider first product states of the form $\rho_{12}=\rho_1 \otimes \rho_2$.
Using Eq.~(\ref{eq:R12ProductState}) one obtains the following:
\beq
\begin{split}
{\mathcal P}(12)&={\mathcal R}(\rho_{1 2}^{T_{2}}){\mathcal R}(\rho_{2 1}^{T_{1}})
={\mathcal R}(\rho_{1 2}^{T_{2}}){\mathcal R}(\rho_{1 2}^{T_{2}})^{\dagger}\\&=
\rho_1^R(\rho_2^R)^{\dagger} \rho_2^R  (\rho_1^R)^{\dagger}=
\mbox{tr}(\rho_2^2) \rho_1^R (\rho_1^R)^\dagger
\end{split}
\eeq 
which is of rank-1 and hence $R_{12}=0$. 

For states of the form $\rho_{12}=\sum_{k=1}^M p_k \rho_{1k} \otimes \rho_{2k}$, a similar calculation yields
\beq
\label{eq:SepR12}
{\mathcal P}(12)= \sum_{k,l=1}^{M} p_k p_l \, \mbox{tr}(\rho_{2k} \, \rho_{2l}) \rho_{1k}^R (\rho_{1l}^R)^\dagger.
\eeq
This cannot be of full rank if $M<d_1^2$, and hence $R_{12}=0$ in this case as well.

\item  {\it Classical-Quantum correlated states:} These are of the form 
\cite{Piani2014,ShunlongLuo2008,KavanModi2014}
\beq
\label{ClassicalQState}
\rho_{CQ}=\sum_{i} p_{i}\, |i\kt \br i| \otimes \rho_i,
\eeq
where $\{|i\kt \}$ are orthonormal and $\rho_i$ are arbitrary states. In this and the next example we use subscripts $C$ and $Q$
for the subsystems so as to make the classical and quantum labels explicit.
As a special case of separable states with $M \leq d_C < d_C^2$ it follows from Eq. ~(\ref{eq:SepR12}) that $R_{CQ}=0$. 
Alternatively it is straightforward to verify that 
\[ {\mathcal P}(CQ)=\sum_{ij} \mbox{tr}(\rho_i \rho_j)   p_{i} p_{j} |ii\kt \br jj|, \]
and hence there are at least $d_C^2-d_C$ vanishing eigenvalues of  ${\mathcal P}(CQ)$,
implying that $R_{CQ}=0$. 
{
It should be noted that this is true irrespective of whether the dimension of the classical subsystem is 
greater or less than that of the quantum.}

{
\item  {\it Quantum-Classical correlated states:} From the symmetry property of $R$ it may appear that we
can conclude that it vanishes also for Quantum-Classical correlated states where the orthonormal projectors are
in the second subspace. The nonzero eigenvalues of ${\mathcal P}(QC)$ are the same as that of ${\mathcal P}(CQ)$,
which are at most $d_C$ in number. Therefore if $d_Q^2 \leq d_C$ we cannot conclude that ${\mathcal P}(QC)$
is rank deficient, and hence $R_{QC}$ maybe nonzero in this case. On the other hand if  $d_Q^2> d_C$ we can 
indeed conclude that $R_{QC}=0$. This is related to the restriction on the number of product states in the general separable case considered in Eq. ~(\ref{eq:SepR12}).} 

\end{enumerate}

As a simple special case it also follows that $R$ vanishes for all Classical-Classical correlated states.
It is interesting that the quantum discord also vanishes for classical-classical states.
In the case of classical-quantum correlated states it vanishes only when the measurements are done on the classical
part of the state otherwise it does not vanish in general. In fact it vanishes iff the state is of this 
form and measurements are done on the classical part of the state \cite{ShunlongLuo2008}.
The quantity $R_{12}$ is a very simply calculable quantity unlike the quantum discord, and also vanishes for such states.
Thus, $R_{12}$ does not seem to include any correlations that are excluded by quantum discord.

Thus, for rank-1 two-qubit density matrix $R_{12}$ gives entanglement. But for higher ranks
$R_{12}$ contains not just entanglement but correlations of other types also.
{Since $R_{12}$ is non-constant and non-zero for separable states as seen in the case
of separable Werner states, it is not an entanglement monotone \cite{Plenio07,Horodeckirpm}.
But since $R_{12}$ equals concurrence for two-qubit pure states it is an entanglement monotone 
only on such states.
}
In this respect it is similar to the quantum discord 
\cite{ZurekQuantumDiscord2000,ZurekQuantumDiscord2001,KavanModi2012} which includes other correlations and it 
almost never vanishes. 
{Quantum discord too is not an entanglement monotone except on bipartite pure states
since in this case it equals von Neumann entropy.
}
However of course the discord has an information 
theoretic interpretation of a difference of two types of mutual information.  It will be seen below how
the difference between $R_{12}$ and $C_{12}$ for rank-2 two qubit states has  a natural interpretation.

\section{Rank-2 two-qubit density matrices or pure three-qubit states}
\label{sec:ThreeQubitPureStatesChp5}

Any rank-2 two-qubit density matrix can be obtained as a reduced density matrix of a suitable three-qubit pure state.
In this section the measure $R_{12}$ is explored for such rank-2 density matrices with this
intrinsic relationship to three-qubit states in mind.
Three-qubit states (pure as well as mixed) have been actively explored in the recent
past~\cite{Coffman,Carteret00,Acin00,Sudbery01,AcinThreeQubit01,Uhlmann06,UhlmannThreeTangle08,RungtaTripartite09,Vidal00,Eltschka12,Siewert12,Soojoon05,Tamaryan09,Ajoy10}, as they offer the simplest system that includes multiparty entanglement sharing. Construction of LU invariants was discussed in ~\cite{Sudbery01}, however the quantity 
$R_{12}$ was not considered therein. A pure multipartite entanglement measure, the 3-tangle was introduced in ~\cite{Coffman}. The archetypal states of GHZ and W were superposed and resulting entanglement studied in ~\cite{Greenberger89,Vidal00,Uhlmann06,UhlmannThreeTangle08}.  
These  were based on various entanglement measures like concurrence, negativity, log-negativity,
entanglement of formation, von Neumann entropy, squashed entanglement, and so on \cite{Horodeckirpm}. 

The canonical form of three-qubit pure states~\cite{Acin00} is extremely useful and given by
\beq
\label{eq:3stategene}
|\psi\rangle =\lambda_{0}|000\rangle + \lambda_{1} e^{i\theta}|100\rangle + \lambda_{2}|101\rangle 
+ \lambda_{3}|110\rangle + \lambda_{4}|111\rangle,
\eeq
where $\lambda_i \in[0,1]$ $\forall i=0$ to $4$, $\sum_i \lambda_i^2=1$ and $\theta\in[0,\pi]$.
It can be seen that Eq.~(\ref{eq:3stategene}) has only five independent parameters excluding the trace \cite{Sudbery01}.
The general three-qubit state with $16$ real parameters is reduced to these five by LU transforms that have $3\times 3=9$ real parameters, and by accounting for overall normalization and phase. Thus, this is the minimal form of any such state.
The following proposition for a two-qubit density matrix of rank two now follows.
\begin{proposition}{For a rank-2 two-qubit density matrix $R_{12}^2 \leq C_{12} \leq R_{12}$, or equivalently $C_{12}\le R_{12} \le \sqrt{C_{12}}$}. 
\label{prop:RCineq}
\end{proposition}
\begin{proof}
Using the canonical form of a three-qubit state which is a purification of the given state $\rho_{12}$ leads to the 
following parametrization
\begin{equation}
\label{eq:RDM2Qubits}
\rho_{12} =\left( \begin{array}{cccc}
    \lambda_{0}^2 & 0 &  \lambda_{0}\lambda_{1} e^{i\theta} & \lambda_{0}\lambda_{3}\\
	      0 & 0 & 0 & 0\\
	      \lambda_{0}\lambda_{1} e^{-i\theta} & 0 & \lambda_{1}^2+\lambda_{2}^2  & \lambda_{1} \lambda_{3}e^{-i\theta}+\lambda_{2}\lambda_{4} \\     
	      \lambda_{0}\lambda_{3} & 0 &  \lambda_{1} \lambda_{3}e^{i\theta}+\lambda_{2}\lambda_{4} & \lambda_{3}^2+\lambda_{4}^2\\
\end{array}\right).
\end{equation}
The following measures are readily calculated:
\beq
\label{eq:C12R12}
C_{12}=2 \lambda_{0} \lambda_{3} \;\; \mbox{and}\;\; 
R_{12}=2\lambda_{0}\lambda_{3}^{1/2}(\lambda_{3}^2+\lambda_{4}^2)^{1/4}.
\eeq
Considering the difference
\beq
C_{12}^4 - R_{12}^4=-16\lambda_{0}^4 \lambda_{3}^2 \lambda_{4}^2 \le 0,
\eeq
it follows that $C_{12} \leq R_{12}$.

Similarly, 
\beq
\begin{split}
R_{12}^4-C_{12}^2&=4\, \lambda_{0}^2 \lambda_{3}^2 \left[4 \, \lambda_{0}^2(\lambda_{3}^2+\lambda_{4}^2)-1\right] \\ &\le 4\, \lambda_{0}^2 \lambda_{3}^2 \left[4 \, \lambda_{0}^2(1-\lambda_0^2)-1\right] \le 0,
\end{split}
\eeq
where the first inequality follows from $\lambda_{3}^2+\lambda_{4}^2  \le 1-\lambda_{0}^2$, a consequence of the normalization constraint. The second follows from the fact that $x(1-x)$ attains the maximum value of $1/4$ when $0\le x \le 1$. Thus, the proposition is proved.
\end{proof}

The equality $R_{12}=C_{12}$ in the case of rank-1 two qubit states becomes broadened into 
this inequality and the measure $R_{12}$ is never smaller than concurrence in the case of rank-2 states.
However there is also an upper bound on $R_{12}$, as it is never larger than $\sqrt{C_{12}}$. Thus, $C_{12}$ and $R_{12}$ when they vanish, do so simultaneously. 

\subsection{The 3-tangle, concurrence, and $R_{12}$ in three-qubit pure states}
\label{EqualityForConcurrence}

The difference $R_{12}^4-C_{12}^4$ is now shown to be related to the tripartite entanglement in the three-qubit purification, namely the 3-tangle. To recall the definition of the 3-tangle \cite{Coffman}, it is given by 
$\tau = C_{1(23)}^2-C_{12}^2-C_{13}^2$ where $C_{ij}$ is the concurrence between qubits $i$ and $j$. The quantity $C_{1(23)}$ is the concurrence between qubit $1$ and 
the pair of qubits $2$ and $3$, since in the case of three-qubit pure state the reduced density matrix of qubits 
$2$ and $3$ is of rank-$2$.  The 3-tangle $\tau$ has been shown to be permutationally  invariant and $0\le \tau \le 1$. 
The expression of the 3-tangle is rather complicated, and can be expressed as a Cayley hyperdeterminant \cite{Coffman}, however in terms of the parameters of the canonical form it is simply $\tau=4(\lambda_{0}\lambda_{4})^2$. Along with Eq.~(\ref{eq:C12R12}) this leads to the remarkably simple relation
\beq
\label{eq:RCtau}
R_{12}^4=C_{12}^2(C_{12}^2+\tau).
\eeq

It can be seen that in the case of two-qubit rank-one density matrices i.e. two-qubit pure states where 
the 3-tangle ($\tau$) is equal to zero the equality above reduces to the one proved in \cite{Udaysinh13}
namely $C_{12}=R_{12}$.
From the point of purification it implies that if the two quantities $C_{12}$ and $R_{12}$
of a two-qubit density matrix are fixed then after purification to a three-qubit pure state
the 3-tangle ($\tau$) of the final state also gets fixed and is given using Eq.~(\ref{eq:RCtau}) as follows:  
\beq
\tau=\dfrac{R_{12}^4 - C_{12}^4}{C_{12}^2}.
\eeq
It should be noted that the 3-tangle is permutation invariant, and hence the combination in the right hand side 
will inherit this property.

\begin{figure}[t!]
\begin{center}     
\includegraphics[scale=0.35]{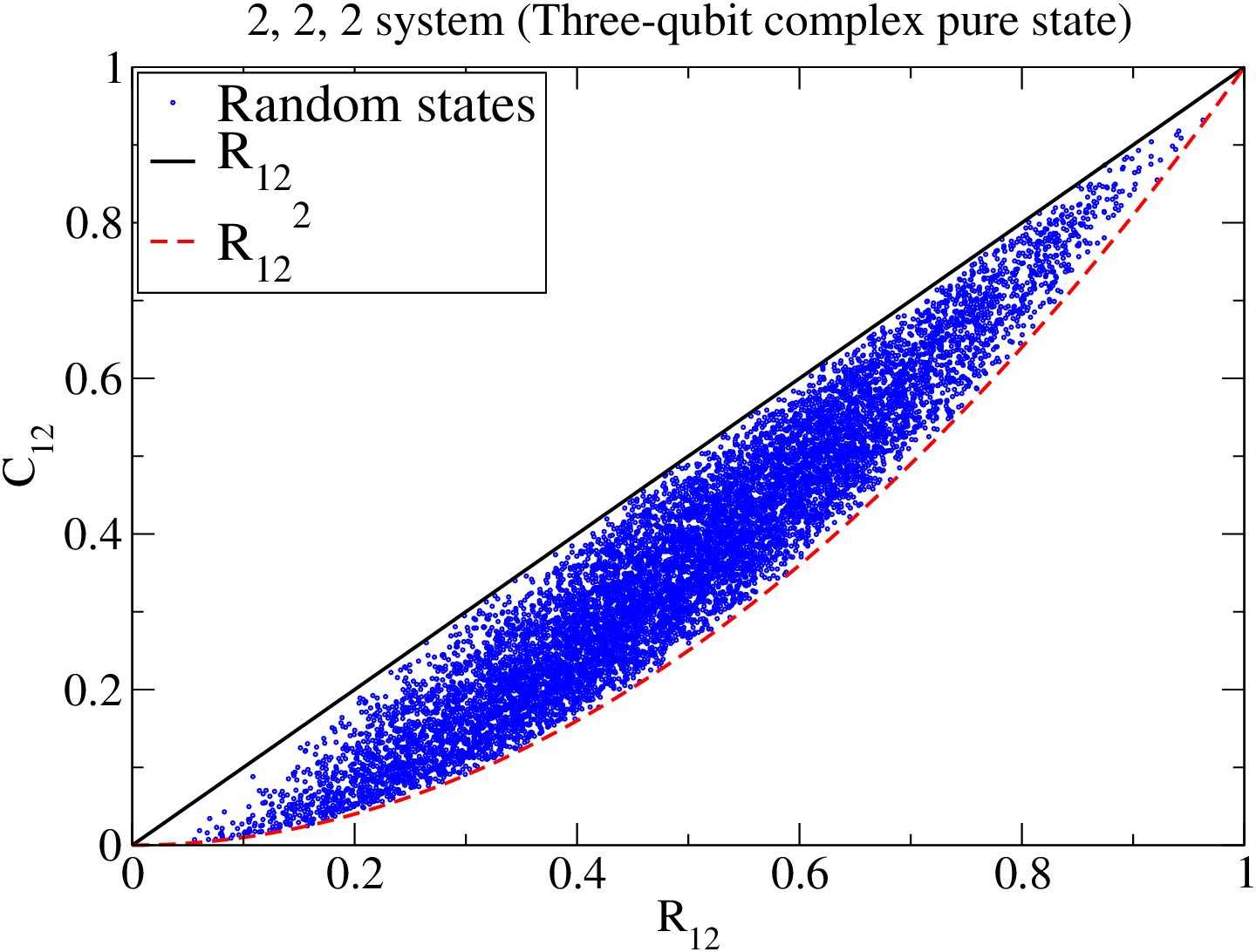}
\caption{(Color online) 
{The $R_{12}$ and the entanglement between two qubits $C_{12}$ having density matrix 
$\rho_{12}$ of rank-2. Also shown are the bounds from proposition $2$, the straight line being W states, while the 
parabola consists of maximally 3-tangled ones ($M3TS$).
Here, $10000$ random tripartite complex three-qubit pure states are used.}} 
\label{fig:inequalityfig1} 
\end{center}
\end{figure} 

\begin{figure}[t!]
\begin{center}     
\includegraphics[scale=0.35]{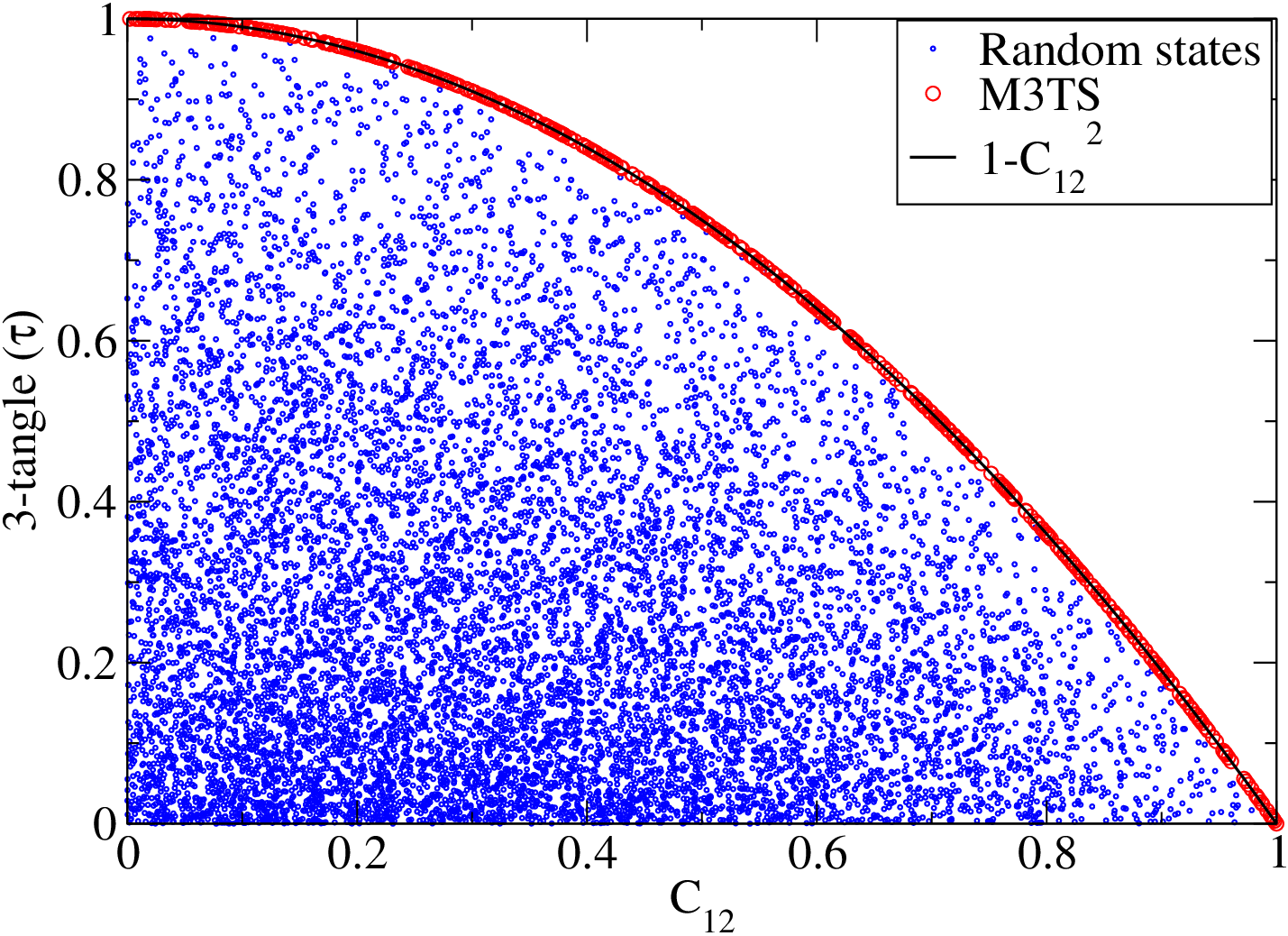}\caption{(Color online) 
{The 3-tangle and the concurrence between first and second qubit $(C_{12})$ of a three-qubit pure 
state is shown. Here, $10,000$ such states are sampled according to the uniform Haar measure. The $M3TS$ states 
are selected from Eq.~(\ref{eq:spestates}) where the value of the parameter $C_{12}$ is chosen randomly. }} 
\label{eq:tanconcdigr} 
\end{center}
\end{figure}

Fig.~\ref{fig:inequalityfig1} shows $R_{12}$ {\it vs}  concurrence between two 
qubits where the overall three-qubit state is randomly sampled  according to the Haar measure. 
The relation between the two measures as reflected in the inequality of Proposition \ref{prop:RCineq} is seen as the region bounding the straight line and the parabola. It is of interest to analyze the states that make up the boundaries of this region.

\subsection{The upper boundary are from $W$ states}
\label{CharacterizingBoundaries}
The upper boundary in Fig.~\ref{fig:inequalityfig1} corresponds to $C_{12} = R_{12}$.
Using Eq.~({\ref{eq:C12R12}) this is seen to imply that  $\lambda_4=0$ when the entanglement is nonzero.
The corresponding states are therefore the $W$ class of states \cite{Vidal00}:
\beq
\label{eq:Chp5Char1}
|\psi\rangle =\lambda_{0}|000\rangle + \lambda_{1} e^{i\theta}|100\rangle + \lambda_{2}|101\rangle 
+ \lambda_{3}|110\rangle.
\eeq 
It is easy to calculate then that 
$ C_{12}=R_{12}=2\lambda_{0}\lambda_{3}, C_{13}=R_{13}=2\lambda_{0}\lambda_{2}, C_{23}=R_{23}=2\lambda_2 \lambda_3$
The tripartite measure of entanglement, the 3-tangle simply denoted as $\tau$ is $0$ for all these states. Note that the relation in Eq.~(\ref{eq:RCtau}) implies that for the upper boundary indeed $\tau=0$. Thus, the $W$ class of states are similar to pure two-qubit states, in as much as there is no difference between the measure $R_{12}$ and concurrence.  $W$ states maximize the concurrence between two qubits for a given $R_{12}$. This is reminiscent of states that maximize concurrence for a given purity, and indeed these are seen to be on the upper boundary too for the following reason.
 
\subsubsection{MEMS~I is the reduced density matrix of $W$ state}
\label{sec:MEMS1isWstate}

It will be now shown that the MEMS~I are  reduced density matrices
of $W$ class states. 
In the section~\ref{sec:SimpleLUI} it is shown that $R_{12}=C_{12}$ for MEMS I.
Note that while the state 
above is an MEMS~I state only in the range $2/3\leq C_{12}\leq 1$, our use of it is the entire range 
$0 \leq C_{12} \leq 1$. Therefore this state is continued to be referred to as ``MEMS~1" for convenience 
although strictly speaking it is a generalization that is no more maximally entangled when $C_{12}<2/3$.

Using Eq.~(\ref{eq:RCtau}) it can be seen 
that if purification of MEMS~I is carried out, in this case to a three-qubit pure state, 
then the three-tangle of the purified state is equal to zero.
In the earlier section~\ref{CharacterizingBoundaries} it was shown that $C_{ij}=R_{ij}$ for all pairs in $W$ class
of states having zero three-tangle. This naturally raises the question of whether MEMS~I are the 
reduced density matrices of $W$ class of states. Indeed it is found that for parameter values of 
$\lambda_0^2=\lambda_3^2=C_{12}/2$ and $\lambda_2^2=1-C_{12}$ in Eq.~(\ref{eq:Chp5Char1}) of $W$ class of states
the reduced density matrix of the first and the second qubit is that of MEMS~I. 
The corresponding canonical form of the $W$ state for which the reduced density matrix of the first and second qubits is MEMS~I is
\begin{equation}
\label{Eq:MEMS1isWstate1}
 |\psi\rangle=\sqrt{\frac{C_{12}}{2}}|000\rangle+\sqrt{1-C_{12}}|101\rangle+\sqrt{\frac{C_{12}}{2}}|110\rangle.
\end{equation}
MEMS~I states appear later in this paper, when negativity is discussed, as a very different boundary.

\subsection{Lower boundary are from maximally 3-tangled states}
\label{section:M3TS}
The lower boundary of Fig.~(\ref{fig:inequalityfig1}) is characterized by states that have  $C_{12} = R_{12}^2$. Fixing the entanglement $C_{12}$ between two qubits which are part of a 3-qubit pure state, these maximize $R_{12}$. While this seems obscure, it is clear from Eq.~(\ref{eq:RCtau}) that for a given $C_{12}$, maximizing $R_{12}$ is the same as maximizing the
tripartite entanglement as measured by the 3-tangle. In this sense the states that make up the lower boundary are maximally 3-tangled states 
($M3TS$).
Using Eq.~(\ref{eq:C12R12}) one derives that $C_{12} = R_{12}^2$ implies, provided $\lambda_3 \neq 0, \, \lambda_0 \neq 0$,
\beqa
\label{eq:rell3l4}
\lambda_{0}^4-\lambda_{0}^2 \, (1-\lambda_{1}^2-\lambda_{2}^2)+\sfrac{1}{4} &=& 0.
\eeqa
This implies that $\lambda_{1}=\lambda_{2}=0$, else the discriminant of the quadratic equation in $\lambda_0^2$ becomes negative. 
Hence
\[ \lambda_{0}=\sfrac{1}{\sqrt{2}}, \; \lambda_{4}=\sqrt{\sfrac{1}{2}-\lambda_{3}^2}.\]

Thus, only one variable, say $\lambda_3$, is needed to parametrize the lower boundary. 
The pairwise concurrences in such states are $C_{12}=\sqrt{2} \, \lambda_{3}, \;C_{13}=0, \;\mbox{and}\; C_{23}=0.$
The states, using instead of $\lambda_3$, the entanglement $ C_{12}$, are given by
\beq
 \label{eq:spestates}
|\psi_{M3TS}\kt =\dfrac{1}{\sqrt{2}} \left(|000\rangle + C_{12} |110\rangle +\sqrt{1-C_{12}^2} |111\rangle  \right).
\eeq
These are explicit forms of maximally 3-tangled states. If $C_{12}=0$ this results in the three-qubit GHZ states. For  $C_{12}=1$ it reduces to $\left(|00\rangle +|11\rangle \right)|0\rangle/\sqrt{2}$, wherein qubits $1$ and $2$ are maximally entangled and qubit $3$ is not entangled to them. 

The 3-tangle for $|\psi_{M3TS}\kt$ states are 
\beq
\label{eq:M3TStau}
\tau=1-C_{12}^2 = 1-R_{12}^4.
\eeq
In Fig.~\ref{eq:tanconcdigr} the 3-tangle $\tau$ is plotted as a function of entanglement between qubits $1$
and $2$ ($C_{12}$) for a random sampling of three-qubit states. It is clear that the states in Eq.~(\ref{eq:spestates}) gives the maximum 3-tangle for the given value of $C_{12}$. Also as a consequence of this maximization it is found that $C_{13}=C_{23}=0$ for these states, which is a reflection of the monogamy of entanglement. When the entanglement between two qubits in a tripartite system is held fixed, maximizing the multipartite entanglement results in the other two pairs not being entangled.

For these states it also holds that $R_{23}=R_{13}=0$. Indeed the reduced density matrices are 
\beq
\rho_{13}=\rho_{23}=\dfrac{1}{2}\left(|1\rangle \langle 1| \otimes |\alpha \kt \br \alpha |+
|0\rangle \langle 0| \otimes |0\rangle \langle 0|\right),
\eeq
where
\[ |\alpha\kt =C_{12}|0\rangle + \sqrt{1-C_{12}^2} |1\rangle .\]
It is quite clear that these are classical-quantum correlated states as in Eq.~(\ref{ClassicalQState})
and it has been shown already in the section~\ref{sec:SimpleLUI} that $R_{13}$ and $R_{23}$ 
indeed vanishes for this class. 
{
On the other hand the reduced density matrix of the first and the second qubit is  a mixture of two Bell states
\beq
\rho_{12}=\left(\dfrac{1+C_{12}}{2}\right)|\phi^{+}\rangle\langle\phi^{+}|+
\left(\dfrac{1-C_{12}}{2}\right)|\phi^{-}\rangle\langle\phi^{-}|,
\eeq
therefore a special case of Bell-diagonal states \cite{Wootersentform,Verstraete2001}.
}

A natural generalization of $|\psi_{M3TS}\kt$ presents itself as  states whose 3-tangle is maximized under constraints that 
two of the pair entanglements are fixed, say $C_{12}$ and $C_{13}$. Stationary points of the 3-tangle with these constraints lead to ($\lambda_0=1/\sqrt{2},\lambda_1=0,\lambda_2=C_{13}/\sqrt{2},\lambda_3=C_{12}/\sqrt{2}$) 
and can be shown to be a maxima.
Thus, the states have two parameters and are given as follows:
\begin{equation}
\label{eq:spestates2}
\begin{split}
|\psi\rangle=\dfrac{1}{\sqrt{2}}\bigg(|000\rangle+ C_{13} |101\rangle 
+C_{12} |110\rangle+ \\ \sqrt{1-C_{12}^2-C_{13}^2} |111\rangle\bigg),
\end{split}
\end{equation}
and the 3-tangle $\tau$ of these states is given as follows:
\beq
\label{eq:ThreeTangleM3TS2}
\tau=1-C_{12}^2-C_{13}^2.
\eeq
Thus, the coefficients $C_{12}$ and $C_{13}$ have to be chosen such that $0\leq C_{12}^2+C_{13}^2 \leq 1$.
Here, $C_{12}=\sqrt{2}\lambda_3$, $C_{13}=\sqrt{2}\lambda_2$ 
and interestingly it is found that $C_{23}=C_{12}C_{13}$.
The states in Eq.~(\ref{eq:spestates2}) are natural generalizations of the $M3TS$.
Contrary to the $M3TS$ case, if the entanglement (concurrences) between first and second qubit, and between first 
and third qubit are kept fixed at some non-zero value then the maximization of the 3-tangle does not 
lead to zero entanglement (concurrence) between the second and third qubit.
For these states it can be seen that $R_{12}=\sqrt{C_{12}}(1-C_{13}^2)^{1/4}$,
$R_{13}=\sqrt{C_{13}}(1-C_{12}^2)^{1/4}$ and $R_{23}=\sqrt{C_{12}C_{13}}(1-C_{12}^2)^{1/4}(1-C_{13}^2)^{1/4}$, so that $R_{23}=R_{12}R_{13}$ as well.
These expressions reduces to those of $M3TS$ for $C_{13}=0$. For other interpretations and appearance of the $M3TS$ please see the last section.

\section{Two qubit states with rank $>2$}
\label{sec:2qubithigherrank}

For rank-1 and rank-2 states, or pure two qubit and pure three qubit states the connections between the measure 
$R_{12}$ derived from the partial transpose followed by realignment and standard entanglement measures such a 
concurrence, and 3-tangle is striking and instructive. It also leads naturally to a segregation of the $W$ states, 
as well as to a class of states that are in an essential sense maximally tangled, the $M3TS$ states. When we step 
beyond to rank-3 and rank-4 states, the picture predictably gets murkier. The purifications are for example to 
systems of two qubits and a qutrit for the rank-3 case. Equivalents of 3-tangles are to our knowledge not readily 
available. Yet from lower rank cases we can expect that the difference of some powers of  $R_{12}$ and $C_{12}$ may 
reflect multiparty entanglement present in such purifications. This expectation is predicated on the inequality 
$R_{12} \ge C_{12}$ continuing to hold for rank-3 and rank-4 cases, which appears to be true.

{This was checked numerically in two ways. One, by direct Haar sampling of pure states of two qubits and one 
qutrit (rank-3 case) or one ququad (rank-4 case), and constructing $R_{12}$ and $C_{12}$.} The result is shown in Fig.~(\ref{fig:inequality34}) where only 
rank-3 and rank-4 cases are shown. It is clear that the straight line corresponding to $C_{12}=R_{12}$ is being 
typically avoided, and in fact $R_{12}$ {\it vs} $C_{12}$ is quite large. States close to the equality line are therefore all 
rank-1 and some rank-2 states. This is in accordance with the picture that emerged of the difference being a 
multiparty entanglement of the purification. We expect that there is more of such entanglement present in 
purifications of rank-3 or rank-4 states. Note that the latter can also be thought of reduced density matrices of 
four qubit pure states.
\begin{figure}[t!]
\begin{center}     
\includegraphics[scale=0.35]{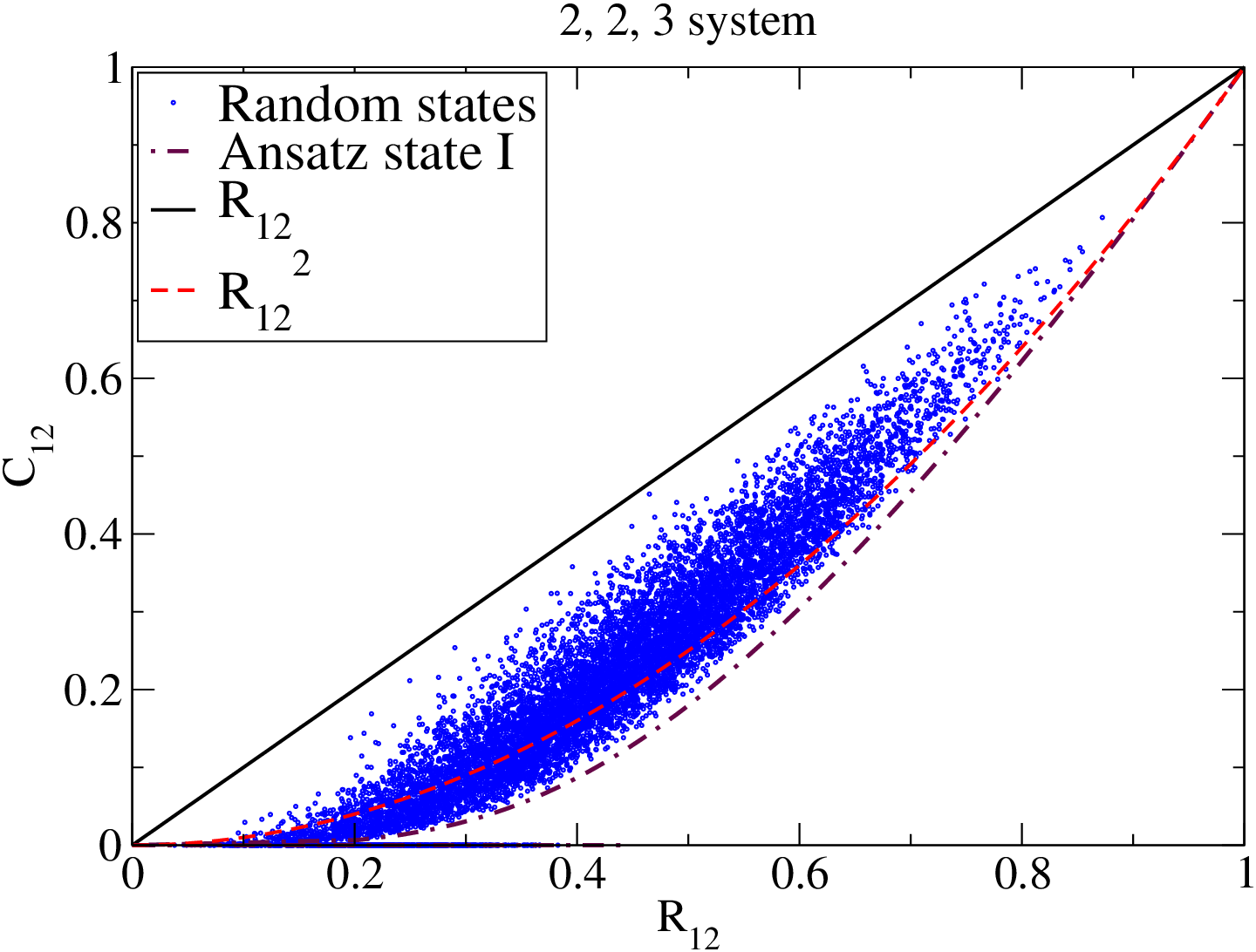}
\includegraphics[scale=0.35]{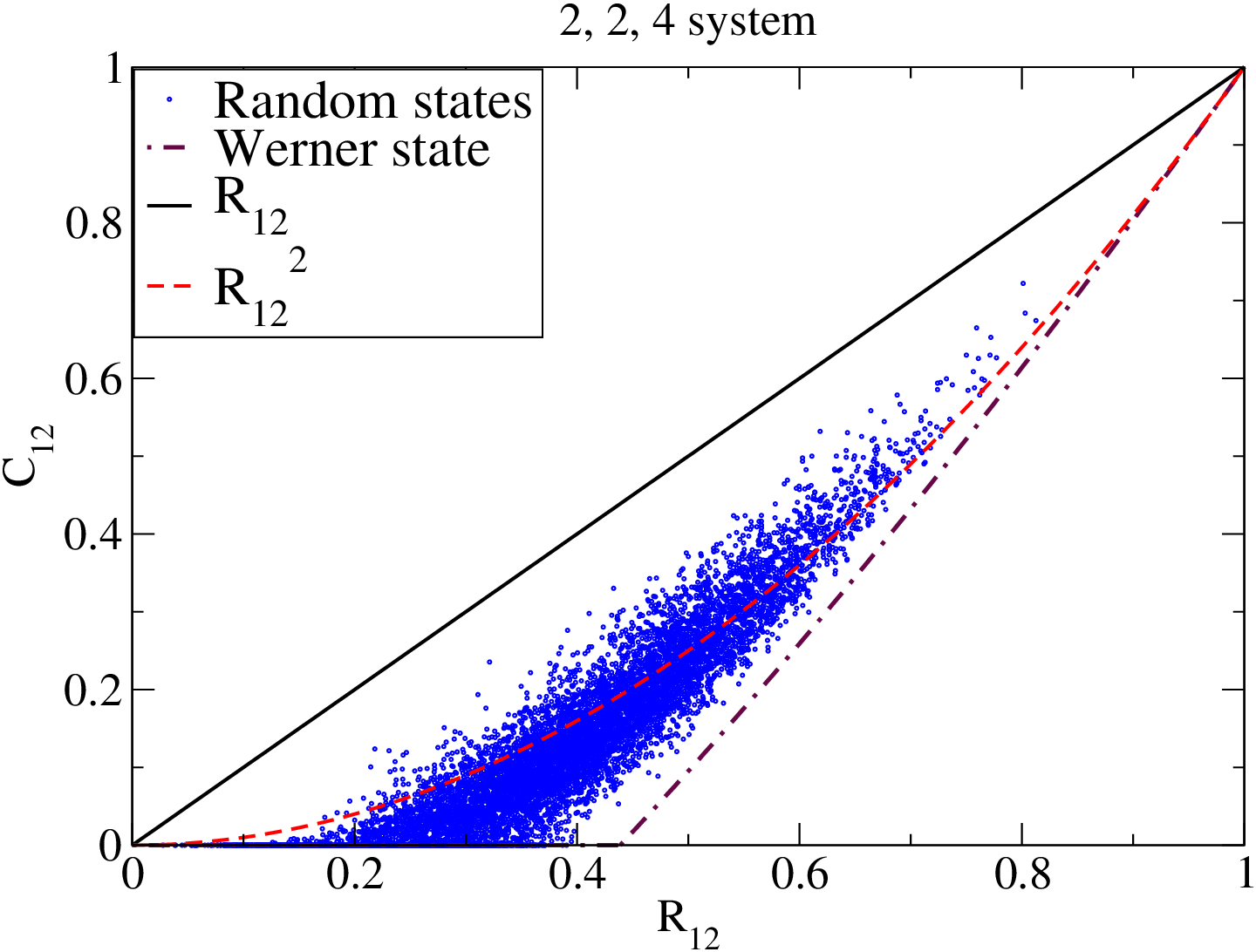}
\caption{(Color online) 
{The $R_{12}$ and the entanglement between two qubits $C_{12}$ having density matrix 
$\rho_{12}$ of rank-3 and -4 at top and bottom respectively. Also shown are the boundaries for rank-2 states for comparison. Ansatz-state-I  curve in top figure corresponds to Eqs.~(\ref{Eq:R12C12Rank3Second}) 
and (\ref{Eq:R12C12Rank3}). Werner state curve in the bottom figure corresponds to Eq.~(\ref{Eq:R12C12Rank4Ineq}).
Here, $10000$ random tripartite complex pure states with respective dimension of the third 
subsystem are used.}} 
\label{fig:inequality34} 
\end{center}
\end{figure}

The second method by which this was checked was by explicitly constructing the density matrix $\rho_{12}$ (real case) in its spectrally decomposed form, where its four eigenvectors and eigenvalues are parameterized.
Then, with this parameterized eigenvalues and eigenvectors, and using Mathematica~$9$ it is checked that 
the density matrix $\rho_{12}$ satisfies $R_{12} \ge C_{12}$. 
Thus, it is proved for real rank-4 two-qubit density matrices 
$R_{12} \ge C_{12}$ and it is 
most likely valid for complex states as well. Considering arbitrary rank-4 perturbations of pure states also validated this inequality.
{Further, in this regard see the comment at the end of the paper.}

{

\subsection{States in the lower boundary of the rank-3 case}
\label{sec:AnsatzC12R12Rank3}

States that give the lower boundaries in 
Fig.~\ref{fig:inequality34} for rank-3 and -4 are now discussed. They are assumed to be mixtures of Bell states as was the lower boundary of the rank-2 case, namely the reduced density matrix of $M3TS$ states.
Consider first the case of rank-3, for which an ansatz for
the states in the lower boundary is 
\beq
\label{Eq:AnsatzRank3}
\rho_{12}=\dfrac{1-p}{2}\,\left(|\psi^{+}\rangle\langle\psi^{+}|+|\psi^{-}\rangle\langle\psi^{-}|\right)
+p\,|\phi^{+}\rangle\langle\phi^{+}|.
\eeq
This will be referred to as ``ansatz-state-I" in the following.
It is readily verified that 
\beq
\label{Eq:C12pR12pRank3}
C_{12}=\mbox{max}\{0, 2p-1\},\;\;R_{12}=\sqrt{p}\, |2p-1|^{1/4}.
\eeq

For $0\leq p\leq 1/2$ the concurrence is zero, however $R_{12}$ increases from $0$ to a maximum value of 
$(1/3)^{3/4}$ when $p=1/3$, and decreases again to zero at $p=1/2$. Thus, there are two distinct segments in 
the $R_{12}$ {\it vs}  $C_{12}$ graph for this state, one a horizontal segment 
\beq
\label{Eq:R12C12Rank3Second}
C_{12}=0 \;\;\mbox{for} \;\; 0\leq R_{12}\leq (1/3)^{3/4}.
\eeq
and the other, the curve 
\beq
\label{Eq:R12C12Rank3}
R_{12}= C_{12}^{1/4} \sqrt{\dfrac{1+C_{12}}{2}},
\eeq
when $1/2<p\leq 1$. Both these segments are shown in Fig.~\ref{fig:inequality34} where the regions  
\beq
\label{Eq:R12C12Rank3Ineq}
\begin{split}
&C_{12} < R_{12} \leq C_{12}^{1/4} \sqrt{\dfrac{1+C_{12}}{2}}\;\; \mbox{for}\;\; C_{12} > 0, \\ &\mbox{and}\;\; 0\leq R_{12} \leq (1/3)^{3/4} \;\; \mbox{for}\;\; C_{12} = 0.
\end{split}
\eeq 
are populated. 

Apart from this figure, strong evidence that the lower boundary is indeed given by states in 
Eq.~(\ref{Eq:AnsatzRank3}) is checked numerically by adding a random density matrix $\rho_{R}$
to this ansatz such that the resultant state is still of rank-3.
This random density matrix is such that its eigenvalues are chosen randomly with eigenvectors as $|\psi^{\pm} \rangle$ and  
$|\phi^{+}\rangle$ such that the final density matrix is of rank-3 with unaltered eigenvectors. The
eigenvalues are selected as $\cos^2(\theta)$, $\sin^2(\theta)\cos^2(\phi)$ and $\sin^2(\theta)\sin^2(\phi)$ where 
$\theta$ and $\phi$ are independent random variables chosen uniformly from $[0, \pi]$ and $[0, 2\pi]$ respectively.
The final density matrix $\rho_{12}'$ is given as follows:
\beq
\label{Eq:AnsatzPlusRandom}
\rho_{12}'=\dfrac{\rho_{12}+\varepsilon\rho_{R}}{1+\varepsilon},
\eeq
where $\rho_{12}$ is the ansatz state and $\varepsilon$ is the perturbation parameter which controls the amount of 
randomness added to the ansatz state $\rho_{12}$. 
In Fig.~\ref{fig:AnsatzRandom} results are shown for states 
generated as per Eq.~(\ref{Eq:AnsatzPlusRandom}) for various values of $\varepsilon$. It can be seen from the 
figure (and was also checked numerically) that there is no violation of  Eqs.~(\ref{Eq:R12C12Rank3Ineq}). A more general 
rank-3 state which lies close to the subspace spanned by $\{|\psi^{\pm}\kt , \; |\phi^+ \kt\}$ may be constructed by Gram-Schmidt orthonormalization. Results not presented here reconfirms that the boundary is populated with states as in Eq.~\ref{Eq:AnsatzRank3}.

\begin{figure}[t!]
\begin{center}     
\includegraphics[scale=0.35]{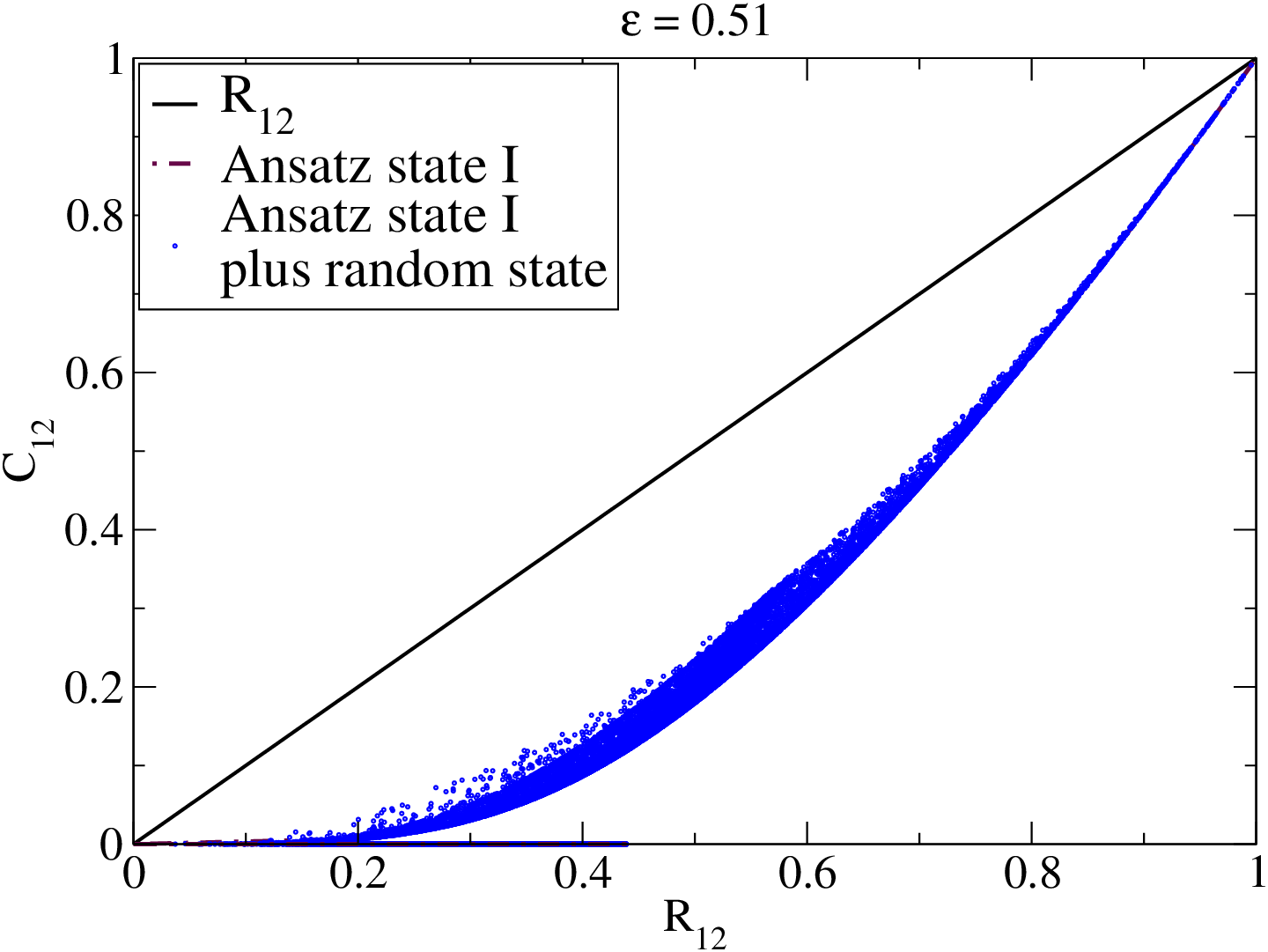}
\caption{(Color online) 
{The $R_{12}$ and the entanglement between two qubits $C_{12}$ having density matrix $\rho_{12}'$ of rank $3$ given 
in Eq.~(\ref{Eq:AnsatzPlusRandom}) for $\varepsilon=0.51$. Here, $20000$ 
such states are sampled randomly in each case. Also shown are the bounds from Eqs.~(\ref{Eq:R12C12Rank3Ineq}). 
}} 
\label{fig:AnsatzRandom} 
\end{center}
\end{figure} 

\subsection{Werner states form the lower boundary in the rank-4 case}
\label{sec:WstateC12method}

Evidence is now presented that Werner states form the lower boundary of rank-4 density matrices, and hence the whole $C_{12}-R_{12}$ diagram. Just as in rank-2 and rank-3 cases, these border states are also Bell-diagonal. Consider the Werner state $\rho_{W}=(1-p) I /4+p \, |\psi^{-}\rangle \langle \psi^{-}|$ \cite{WernerState1989}
($0\leq p \leq 1$), which is entangled iff $1/3 < p \leq 1$ and in that case the entanglement, as measured by the concurrence, is $C_{12}=(3p-1)/2$. As mentioned in the introduction, $R_{12}=p^{3/4}$ for these states. Substituting $p$ in terms of $R_{12}$ in the expression for concurrence one obtains the following:
\beq
\label{Eq:R12C12Rank4Ineq}
C_{12}=\mbox{max}\left\{0,\dfrac{3 R_{12}^{4/3}-1}{2}\right\}.
\eeq
Thus $C_{12}=0$ for $0\leq R_{12} \leq (1/3)^{3/4}$ which is also the case for the ansatz given for rank-3 border states (refer Eq.~(\ref{Eq:R12C12Rank3Ineq})). It is greater than zero for 
$(1/3)^{3/4}<R_{12}\leq 1$, in which case
$R_{12}=\left(\left(2 C_{12}+1\right)/3\right)^{3/4}$.
This curve is plotted in that part of  Fig.~\ref{fig:inequality34} that corresponds to rank-4 states.

To check that Eq.~(\ref{Eq:R12C12Rank4Ineq}) indeed forms the lower boundary for the rank-4 cases, again random perturbations are added 
to the border (Werner) states. First a random tripartite pure is selected
consisting of two qubits and a ququad. The reduced density matrix
of the two qubits $\rho_{R}$ is then added to the Werner state: 
\beq
\label{Eq:WernerRandom}
\rho_{12}=(\rho_{W}+\varepsilon\rho_{R})/(1+\varepsilon)
\eeq

It can be seen from Fig.~\ref{fig:WernerRandom} (and also verified numerically) that there are no states that violate the inequality 
in  Eq.~(\ref{Eq:R12C12Rank4Ineq}). Thus, for the rank-4 case (and hence for generic two-qubit states) it follows that
\beq
\label{Eq:R12C12Rank4IneqM}
C_{12} \leq R_{12} \leq \left(\dfrac{2 C_{12}+1 }{3}\right)^{3/4} \;\; \mbox{for}\;\; C_{12} \geq 0.
\eeq

In Fig.~\ref{fig:AnalyticalCurves} the various boundaries of the $C_{12}-R_{12}$ diagram are shown for arbitrary two-qubit states.
 It includes the parabola from proposition $2$, and curves from Eqs.~(\ref{Eq:R12C12Rank3Ineq})
 and (\ref{Eq:R12C12Rank4IneqM}).
The curve corresponding to rank-1 is the common upper boundary for all the ranks, and is simply the line $C_{12}=R_{12}$.
It can be seen that as the rank increases the corresponding lower boundary gets shifted in the downward 
direction, but always remains a Bell-diagonal state.

}

\begin{figure}[t!]
\begin{center}     
\includegraphics[scale=0.35]{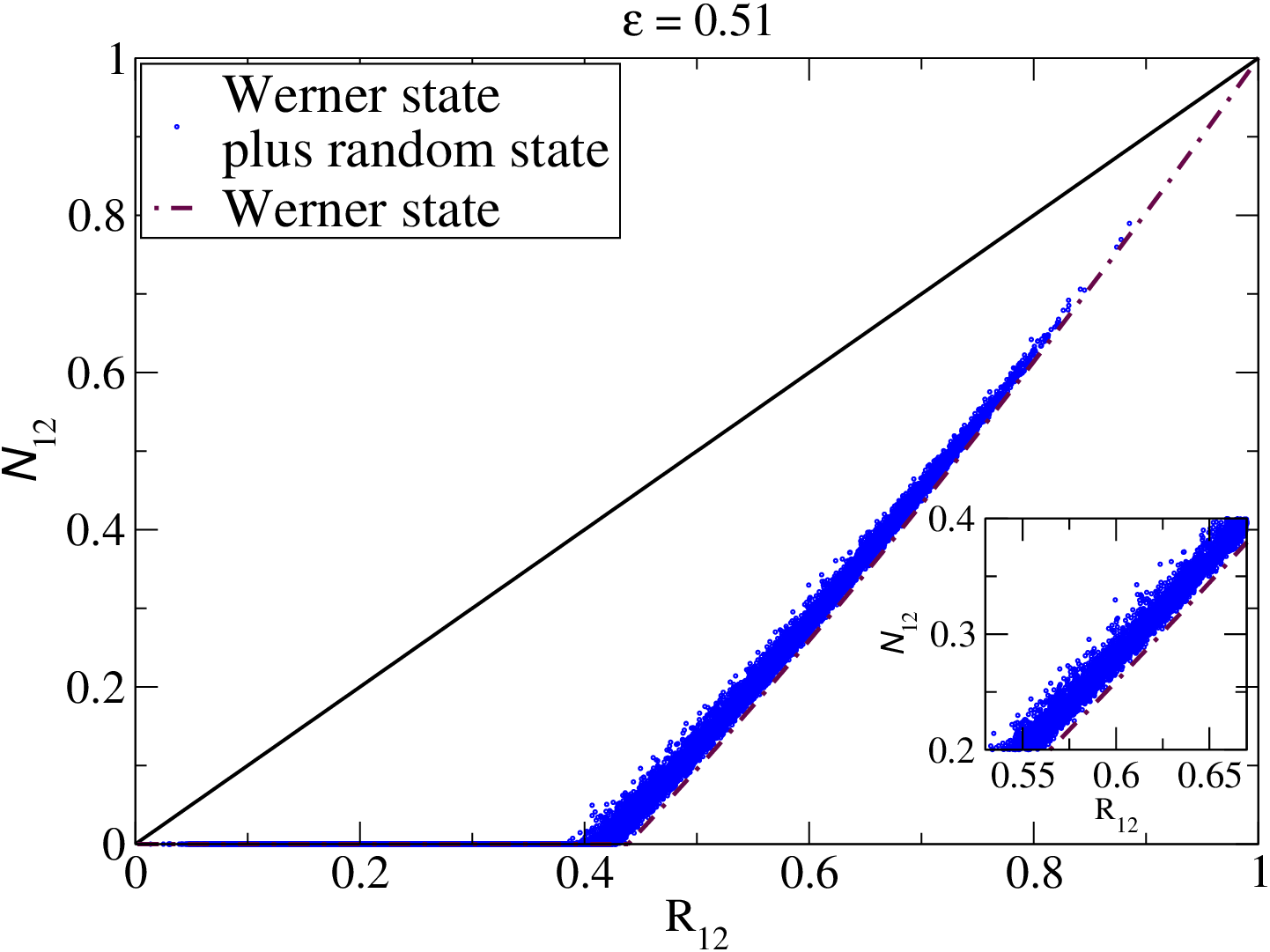}
\caption{(Color online) 
{The $R_{12}$ and the entanglement between two qubits $C_{12}$ having density matrix $\rho_{12}'$ of rank-$4$ given 
in Eq.~(\ref{Eq:WernerRandom}) for values of $\varepsilon=0.51$. Here, $20000$ such states are sampled randomly 
in each case. Werner state curve corresponds to Eq.~(\ref{Eq:R12C12Rank4Ineq}). Also shown are the bounds from 
Eq.~(\ref{Eq:R12C12Rank4Ineq}).
}} 
\label{fig:WernerRandom} 
\end{center}
\end{figure}

\begin{figure}[t!]
\begin{center}     
\includegraphics[scale=0.35]{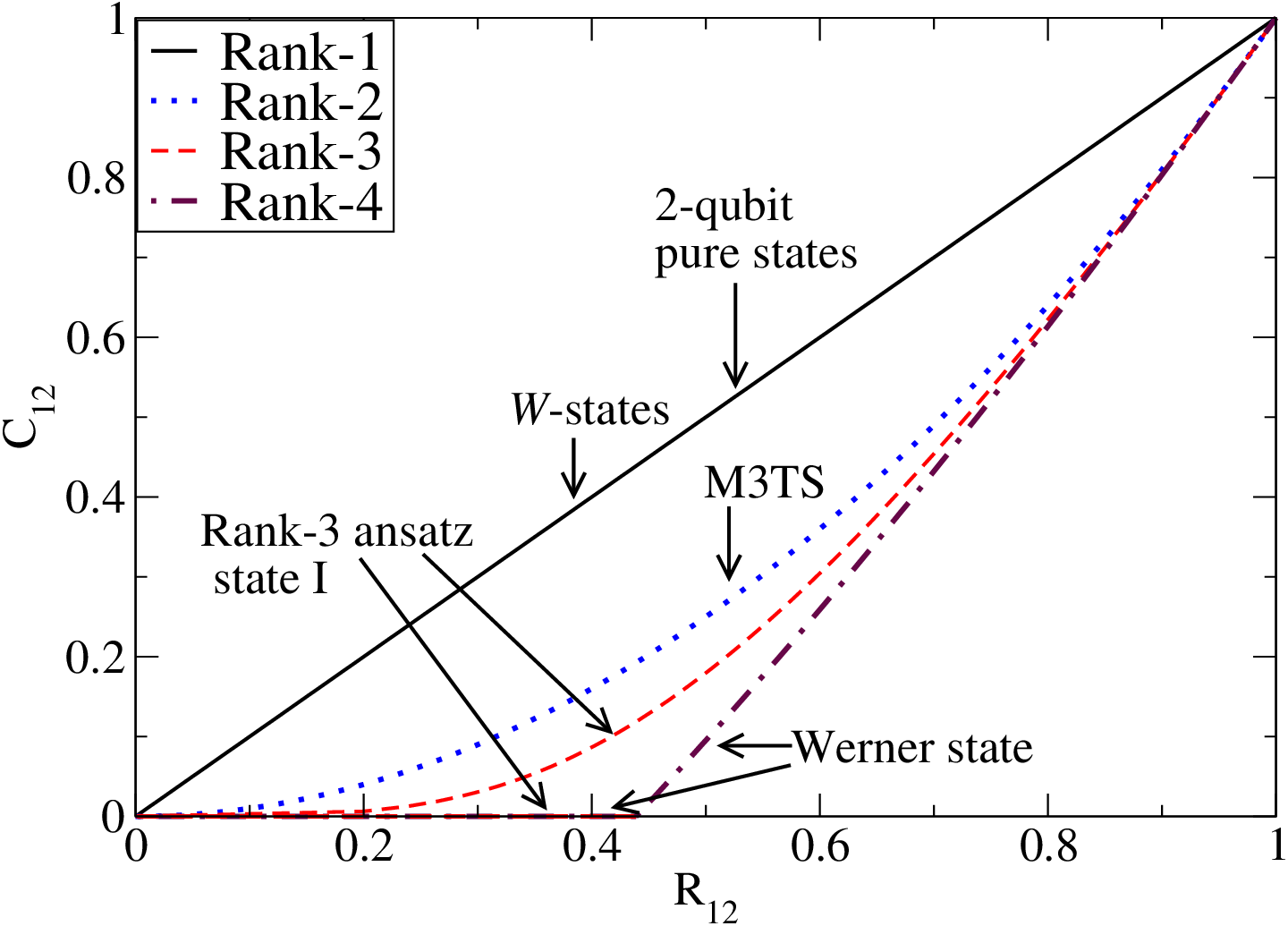}
\caption{(Color online) Boundary curves for $R_{12}$ and the entanglement between two qubits $C_{12}$ having 
density matrix $\rho_{12}$ of all the ranks as per given in proposition $2$, Eqs.~(\ref{Eq:R12C12Rank3Ineq})
and (\ref{Eq:R12C12Rank4IneqM}).
Also shown are various classes of states lying on the respective boundaries.
} 
\label{fig:AnalyticalCurves} 
\end{center}
\end{figure} 

\subsection{Concurrence and $R_{12}$ in $X$ states}
 
An important subset of two qubit states are the so-called $X$ states 
\cite{TingYu2007,ARPRau2009,ARPRau2010,Fanchini2010} which appear in many physical contexts from quantum optics to 
condensed matter \cite{Werlang2010,Sarandy2009,Fanchini2010}. 
They have been intensively investigated, and there are analytical formulas for the quantum discord of the $X$ 
states \cite{ARPRau2010,Fanchini2010}, which were later shown to have very small worst-case error by 
giving explicit counterexamples \cite{YichenHuang2013}.
It is therefore of interest to investigate them in the context of this paper, especially as they are in 
general of full rank. 
 The states are given by the following form that makes the sobriquet ``$X$ states" evident:
\begin{equation}
\label{eq:XStates}
\rho_{X} =\left( \begin{array}{cccc}
	      a & 0 & 0 & w\\
	      0 & b & z & 0\\
	      0 & z^* & c & 0\\
	      w^* & 0 & 0 & d\\
	     \end{array}\right).
\end{equation}
This describes a quantum state provided the unit trace and positivity conditions $a+b+c+d=1$, 
$\sqrt{bc}\geq |z|$, $\sqrt{ad}\geq |w|$ are satisfied. $X$ states are entangled if and only if either
$\sqrt{bc}\le |w|$ or $\sqrt{ad}\le |z|$, and both conditions cannot hold simultaneously \cite{AnnaSanpera1998}.
Concurrence is $C_{12}=2 \,\mbox{max}\{0, |z|-\sqrt{ad}, |w|-\sqrt{bc} \}$.
For $X$ states it is readily seen that 
\beq
R_{12}=2|ad-bc|^{1/4}||z|^2-|w|^2|^{1/4},
\label{eq:R12X}
\eeq
which has an interesting structure, involving the product of the determinants of the reshaped
diagonal and anti-diagonal.

\begin{proposition}{For X states $R_{12}\ge C_{12}.$}
\begin{proof} That 
\[
|ad-bc|^{1/4} \geq ||w|^2-bc|^{1/4}, \; ||w|^2-|z|^2|^{1/4} \geq ||w|^2-bc |^{1/4},
\]
follow from the conditions that $\sqrt{ad}\geq |w|$ and $|z| \leq \sqrt{bc}$, respectively. It follows
then that 
\[
R_{12} \geq 2\, ||w|^2-bc|^{1/2}, \; R_{12} \geq 2\, ||z|^2-ad|^{1/2},
\]
the latter being derived similarly. However it also follows easily that $||w|^2-bc|^{1/2}\geq |w|-\sqrt{bc}$
provided $|w|\geq \sqrt{bc}$ and similarly $ ||z|^2-ad|^{1/2}\geq |z|-\sqrt{ad}$ when $|z| \geq \sqrt{ad}$. 
Thus, whenever the concurrence is nonzero it is necessarily smaller than or equal to $R_{12}.$

\end{proof}
\end{proposition}
As a simple corollary whenever $R_{12}=0$, then $C_{12}$=0. For $X$ states it is clear that $R_{12}=0$ whenever
$ad=bc$ or $|w|=|z|$. It is not evident from the concurrence expressions that it is zero in these cases, but it is so.

{
\section{Negativity and $R_{12}$}
\label{sec:ComR12N12}

While the $R_{12}$-concurrence pair has been exhaustively studied in the case of two-qubit states, it is 
interesting to compare $R_{12}$ with other measures of entanglement. In particular as $R_{12}$ is crucially
dependent on the partial transpose operation, it is of interest to compare it with ``negativity" \cite{Karol1998,Eisert1999,vidal02}
which is exclusively based on the partial transpose operation.

For a two-qubit state $\rho_{12}$, after partial transpose, at most only one eigenvalue can be negative. The negativity is then defined as 
\begin{equation}
N(\rho_{12})=\mbox{max}\{0,-2\mu_{min} \},
\end{equation}
where $\mu_{min}$ is the minimum eigenvalue of the partial transpose of $\rho_{12}$. Unlike concurrence, negativity 
can be calculated for bipartite systems of any dimensionality. If the negativity is non-zero then the state is entangled. 
In that case $\rho_{12}$ is said to be a NPT (negative partial transpose) state otherwise it is a PPT (positive 
partial transpose) state and is guaranteed to be separable only for $2\times2$ and $2\times3$ systems 
\cite{MHorodecki96}. 
Concurrence and negativity for two-qubit states have been previously compared
\cite{Eisert1999,Verstraete2001,Grudka2004} and it was shown that the following holds:
\begin{equation}
\label{eq:InequalityC12N12}
 \sqrt{(1-C_{12})^2+ C_{12}^2}-(1-C_{12}) \leq N_{12} \leq C_{12}.
\end{equation}
It was also shown that the class of states which satisfies the bound $N_{12}=C_{12}$ includes, two-qubit pure states, and 
Bell-diagonal states 
\cite{Verstraete2001,wootters98} (which include Werner states) while the class of states which 
satisfy the lower bound are rank-2 maximally entangled mixed states of rank  {\it i.e.} MEMS~I
\cite{SatoshiIshizaka2000,Verstraete2001,WJMunro01,TzuChiehWei03,NicholasPeters04}.

\subsection{Rank-1 and Rank-2 states}

\begin{figure}[t!]
\begin{center}     
\includegraphics[scale=0.35]{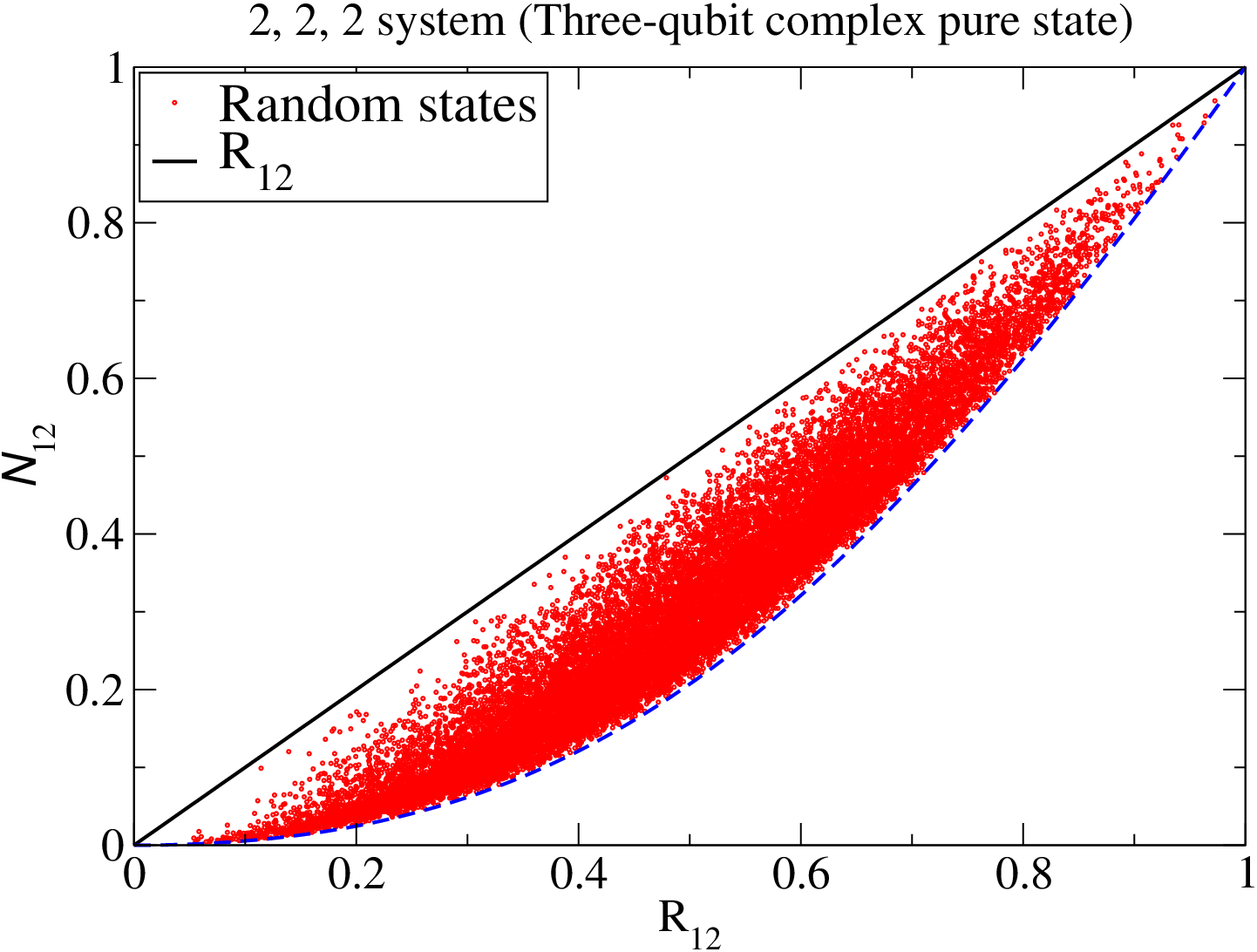}
\caption{(Color online) 
{
The $R_{12}$ and the negativity between two qubits $N_{12}$ having density matrix $\rho_{12}$ of rank-2. 
Also shown are various analytical curves. The lower curve corresponds to the lower boundary in 
Eq.~(\ref{eq:N12R12Curve}).
Here, $20000$ random tripartite complex three-qubit pure states are used.
}} 
\label{fig:222C12Neg12R12} 
\end{center}
\end{figure} 

It should be noted that the inequality in Eq.~(\ref{eq:InequalityC12N12}) holds true for two-qubit density 
matrices of all ranks. However,  as shown below the same inequality, with $R_{12}$ simply replacing $C_{12}$,
 holds true  for the restricted class of two-qubit rank-1 or rank-2 states. That is for these states
 \begin{equation}
\label{eq:N12R12Curve}
 \sqrt{(1-R_{12})^2+ R_{12}^2}-(1-R_{12}) \leq N_{12} \leq R_{12}, 
\end{equation}
To begin with, the inequality in Eq.~(\ref{eq:N12R12Curve}) is checked for hundred thousand random states 
$\rho_{12}$ which are reduced density matrices of pure states of 3 qubits that are drawn from the Haar measure. The bounds are found to hold true in every case, including in a subset  which is shown in 
Fig.~(\ref{fig:222C12Neg12R12}). The boundary states are now presented and verified to be boundaries by 
using perturbations as in the cases above. 

\subsubsection{Pure states have $N_{12}=R_{12}$}
It can be seen easily that for any arbitrary two-qubit pure state $|\psi\rangle=a|00\rangle + b|01\rangle + c|10\rangle + d|11\rangle$,
$N_{12}$ and $R_{12}$ are the same and given by $ N_{12}=R_{12}=2|ad-bc|$. 
Thus, all two-qubit pure states corresponds to the upper boundary of Fig.~\ref{fig:222C12Neg12R12}. It is interesting that while apart from pure states,
$C_{12}=R_{12}$ also for all (rank-2) reduced density matrices from the $W$ class of 3 qubit states, this is no longer the case when negativity is
compared with $R_{12}$. Such states are found not to be special and fill the interior of the region in 
Fig.~\ref{fig:222C12Neg12R12} rather uniformly.

\subsubsection{The lower boundary are MEMS~I}

The $R_{12}$-concurrence lower boundary consisted of $M3TS$ states. These do not form the border
when negativity replaces the concurrence. It is easy to see that for $M3TS$ states $N_{12}=R_{12}^2$. 
MEMS~I \cite{SatoshiIshizaka2000,Verstraete2001,WJMunro01,TzuChiehWei03,NicholasPeters04} were special in the 
$R_{12}$-concurrence case and was the upper boundary as $C_{12}=R_{12}$. The spectral decomposition of MEMS~I 
is 
 \beq
\rho_{MEMS\,I}=(1-C_{12})|01\rangle\langle01|+C_{12}|\phi^{+}\rangle\langle\phi^{+}|
\label{Eq:MEMSIspect}
\eeq
where $R_{12}$ can also be used in place of $C_{12}$ owing to their equality for such states. 
 Again simple explicit calculations yield in this case that $N_{12}=\sqrt{(1-R_{12})^2+ R_{12}^2}-(1-R_{12})$, the lower bound in Eq.~\ref {eq:N12R12Curve}.

 That these states turn out to form the lower 
boundary in the $R_{12}$-negativity case (restricted to rank-2 states) is established by using random rank-2 
perturbations and the result is displayed in Fig.~\ref{fig:222C12Neg12R12M3TSMEMS1}.
The procedure of perturbation is now given. As proved in Sec.~\ref{sec:MEMS1isWstate} that MEMS~I 
is the reduced density matrix of a subset of $W$ class of states as given in Eq.~(\ref{Eq:MEMS1isWstate1}).
Using this 3-qubit purifications of perturbations of MEMS~I can be constructed as  
\beq
\label{Eq:PertMEMS1}
|\phi\rangle=\dfrac{|\psi\rangle+\varepsilon|\psi_R\rangle}{\sqrt{1+\varepsilon^2+\varepsilon\left
(\langle\psi|\psi_R\rangle+\langle\psi_R|\psi\rangle\right)}}.
\eeq
Here $|\psi\rangle$ is a state in the W class restricted to the form in Eq.~(\ref{Eq:MEMS1isWstate1})  ($C_{12}$ uniformly random in $[0,1]$), and $|\psi_R\rangle$ is a three-qubit random pure state selected according to the Haar measure and $\varepsilon$ is the perturbation parameter. Results in Fig.~\ref{fig:222C12Neg12R12M3TSMEMS1} are presented for $\varepsilon=0.51$. It can be seen that the inequality in Eq.~(\ref{eq:N12R12Curve}) is strictly respected. A rather large value of the ``perturbation" is used to clearly show the strict spread of the values to the left of the boundary curve.

\begin{figure}[t!]
\begin{center}     
\includegraphics[scale=0.35]{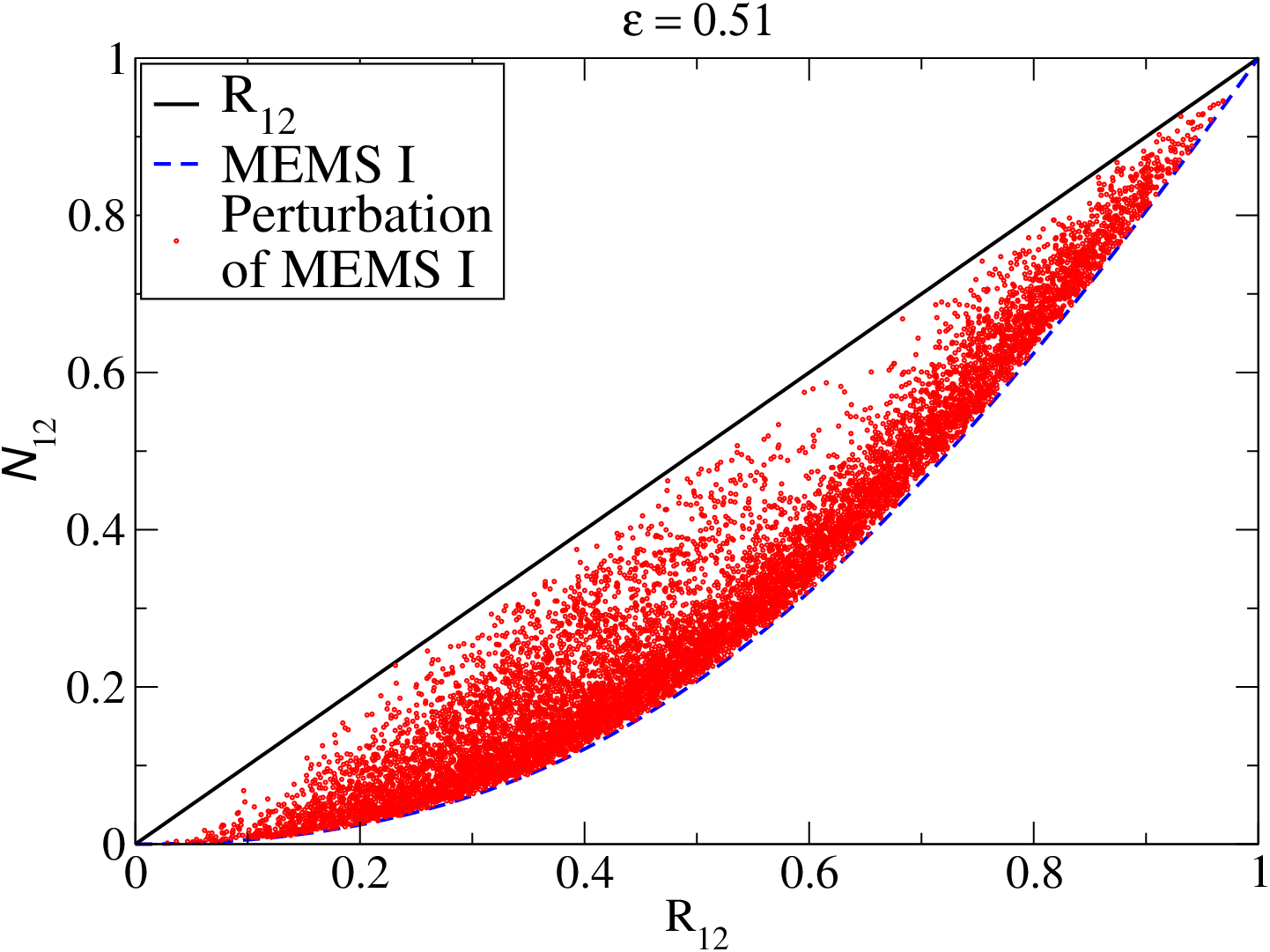}
\caption{(Color online) The $R_{12}$ and the negativity between two qubits $N_{12}$ having density matrix 
$\rho_{12}$ obtained using Eq.~(\ref{Eq:PertMEMS1}). Parameter used is $\varepsilon=0.51$.
Here, $10000$ such states are selected.}
\label{fig:222C12Neg12R12M3TSMEMS1} 
\end{center}
\end{figure}

\begin{figure}[t!]
\begin{center}     
\includegraphics[scale=0.35]{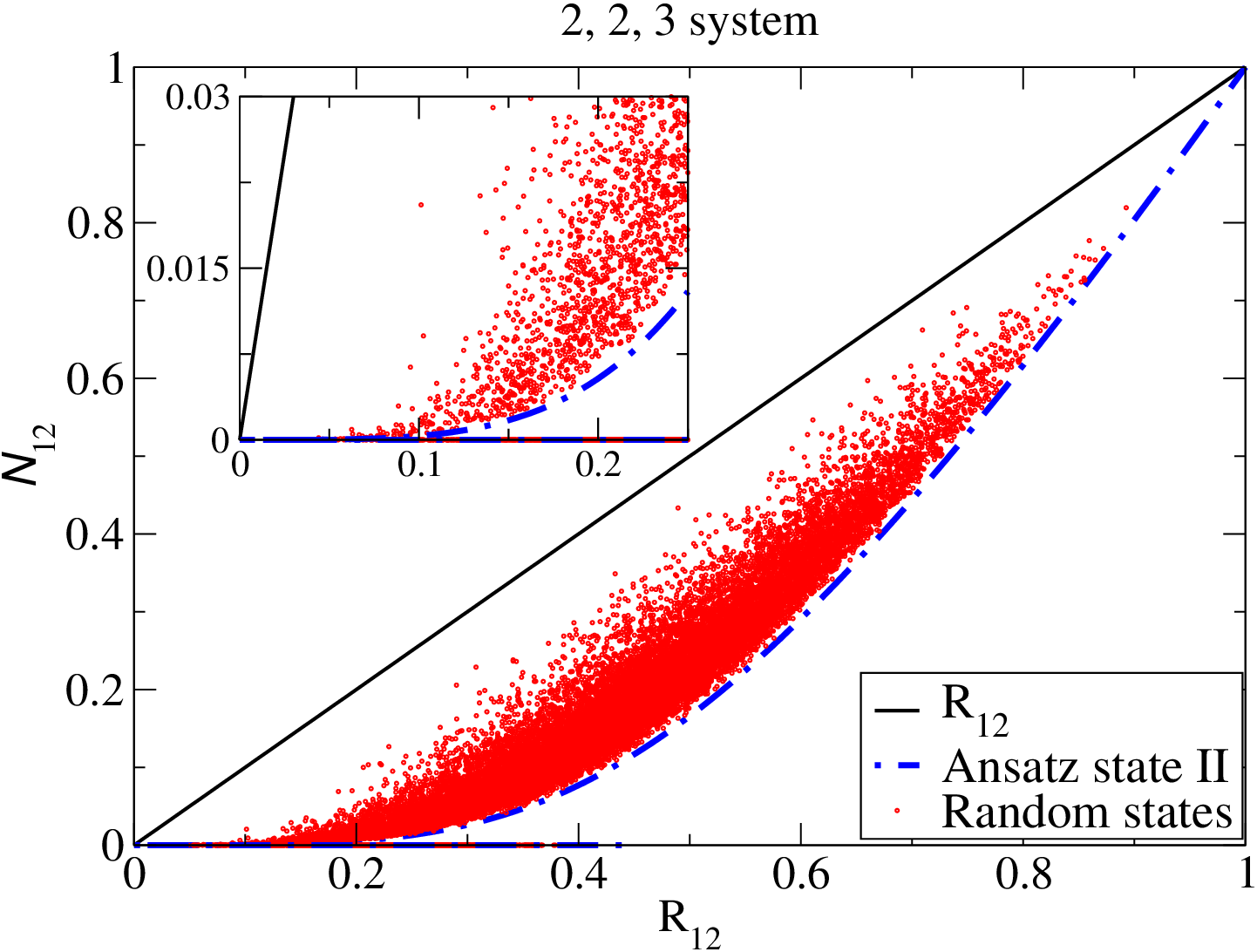}\vspace{1cm}
\caption{(Color online) 
{
The $R_{12}$ and the negativity between two qubits $N_{12}$ having density matrix 
$\rho_{12}$ of rank-$3$. Also shown are various analytical curves. 
Ansatz-state-II in the figure corresponds to Eq.~(\ref{Eq:N12R12CurveRank3}).
Here, $20000$ random tripartite complex pure states with respective dimension of the third 
subsystem are used.
}} 
\label{fig:223224C12Neg12R12} 
\end{center}
\end{figure} 
\subsection{Two qubit states with rank $> 2$}

For rank-3 and rank-4 states it is found from extensive numerical sampling that the lower bound in Eq.~(\ref{eq:N12R12Curve})
is violated while the upper bound is still valid (refer Fig.~\ref{fig:223224C12Neg12R12}).
In other words for rank-3 and rank-4 the inequality $N_{12}\leq R_{12}$ holds. Here too it can be seen that 
there is a lower bound on the spread of the states i.e. for given value of $N_{12}$ there seems to be a  
maximum value taken by $R_{12}$ or for given value of $R_{12}$ there is the minimum value taken by $N_{12}$.

\subsubsection{States in the lower boundary of the rank-3 case}
For the case of rank-$3$ states it is found that neither the ansatz-state-I given in Eq.~(\ref{Eq:AnsatzRank3}) (which gave the
lower boundary for the concurrence {\it vs}  $R_{12}$) nor the MEMS~II define the lower boundary of 
negativity {\it vs}  $R_{12}$. Based on the fact that the spectral decomposition of MEMS~I states which form the lower
boundary of rank-2 states are mixtures of a pure separable state and an orthogonal maximally entangled state (Eq.~(\ref{Eq:MEMSIspect}))
the following  generalization (``anstaz-state-II") is examined:

\beq
\label{Eq:AnsatzRank3N12R12}
\rho_{12}=\alpha|01\rangle\langle01|+\beta|\phi^{+}\rangle\langle\phi^{+}|+\gamma|\phi^{-}\rangle\langle\phi^{-}|
\eeq
where $0\leq\alpha,\beta,\gamma\leq 1$ and $\alpha+\beta+\gamma=1$.
It is readily verified that for these states
\beq
R_{12}=\sqrt{|\beta^2-\gamma^2|},\;\;N_{12}=\sqrt{\alpha^2+(\beta-\gamma)^2}-\alpha.
\label{eq:R12N12rank3negR}
\eeq

The coefficients $\alpha$, $\beta$ and $\gamma=1-\alpha-\beta$ are now fixed such that for a given value of $R_{12}$ 
the minimum value of negativity is obtained.
Using the method of Lagrange multipliers 
with $\alpha$ as an independent parameter this minimization fixes the other coefficients as
\beq
\begin{split}
\beta&=\dfrac{1}{2}(1-\alpha +\sqrt{(1-\alpha)(1-3\alpha)}),\\ 
\gamma&=\dfrac{1}{2}(1-\alpha -\sqrt{(1-\alpha)(1-3\alpha)}).
\end{split}
\eeq
It should be noted that $ 0\leq \alpha \leq 1/3$ to have valid values of $\beta$ and $\gamma$.
With these values of the coefficients, the state in Eq.~(\ref{Eq:AnsatzRank3N12R12}) is the ansatz-state-II 
for which using Eq.~(\ref{eq:R12N12rank3negR}) one gets
\beq 
R_{12}=(1-\alpha)^{3/4}(1-3\alpha)^{1/4},\;\; N_{12}=1-3\alpha.
\eeq
Hence the corresponding negativity {\it vs} $R_{12}$ curve is
\beq
\label{Eq:N12R12CurveRank3}
R_{12}=N_{12}^{1/4} \left(\dfrac{2+N_{12}}{3}\right)^{3/4},
\eeq
which is indeed found to be the lower boundary in Fig.~(\ref{fig:223224C12Neg12R12}).

As for any two qubit state $N_{12}=0$ if and only if $C_{12}=0$ it follows that the ansatz-state-I given in Eq.~(\ref{Eq:AnsatzRank3}),
with $0 \le p \le 1/2$, belong to the horizontal segment with $N_{12}=0$ in Fig.~\ref{fig:223224C12Neg12R12}. As in that case
the maximum value of $R_{12}$ for rank-3 states with $N_{12}=0$ is $(1/3)^{3/4}$. To summarize the regions 
\beq
\label{Eq:R12N12Rank3Ineq}
\begin{split}
&N_{12} < R_{12} \leq N_{12}^{1/4} \left(\dfrac{2+N_{12}}{3}\right)^{3/4} \;\; \mbox{for}\;\; N_{12} > 0, \\ &\mbox{and}\;\; 0\leq R_{12} \leq (1/3)^{3/4} \;\; \mbox{for}\;\; N_{12} = 0
\end{split}
\eeq 
are populated for rank-3 states. It is checked numerically (results not presented) that perturbing the boundary states 
always results in values of negativity and $R_{12}$ that lie in the interior of this region.
Thus although the properties of MEMS~I motivated the form of the ansatz- state-II for non-zero negativity in the 
rank-3 case, it should be noted that this is different from the rank-3 MEMS~II 
\cite{SatoshiIshizaka2000,Verstraete2001,WJMunro01,TzuChiehWei03,NicholasPeters04}.

\begin{figure}[t!]
\begin{center}     
\includegraphics[scale=0.35]{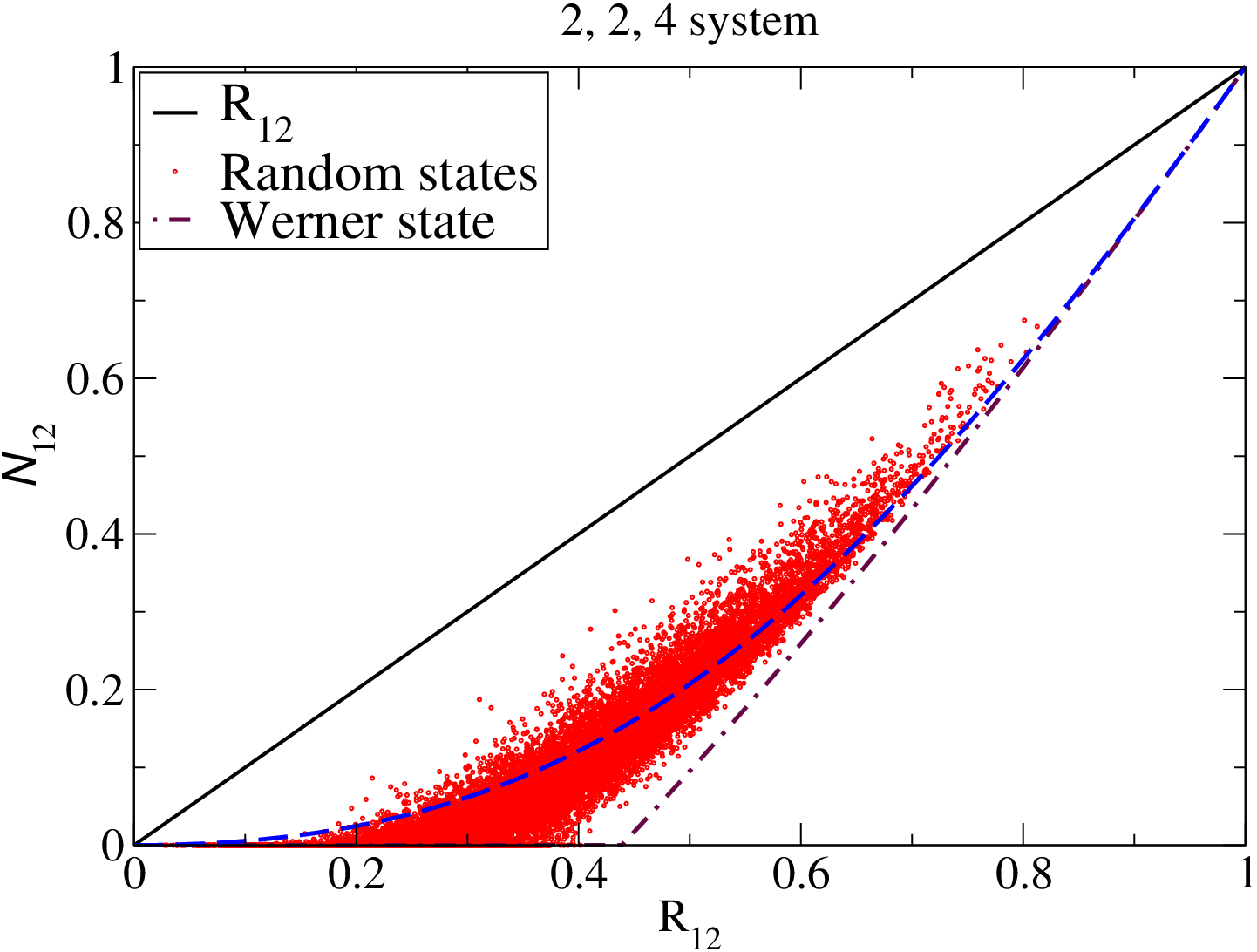} 
\caption{(Color online) 
{The $R_{12}$ and the negativity between two qubits $N_{12}$ having density matrix 
$\rho_{12}$ of rank-$4$. Also shown are various analytical curves. 
The dashed curve in the figure corresponds to the lower boundary in Eq.~(\ref{eq:N12R12Curve}).
The Werner state curve in the bottom figure corresponds to Eq.~(\ref{Eq:R12N12Rank4Ineq}).
Here, $20000$ random tripartite complex pure states with respective dimension of the third 
subsystem are used.
}} 
\label{fig:WernerNeg12R12} 
\end{center}
\end{figure}

\subsubsection{Werner states are in the lower boundary of rank-4 cases}

Evidence is now presented that Werner states form the lower boundary of full-rank states.
As the Werner state is a Bell-diagonal state, $C_{12}=N_{12}$ \cite{Verstraete2001,wootters98},  and it follows on using Eq.~(\ref{Eq:R12C12Rank4Ineq}) that
\beq
\label{Eq:R12N12Rank4Ineq}
N_{12}=\mbox{max}\left\{0,\dfrac{3 R_{12}^{4/3}-1}{2}\right\}.
\eeq
It can be seen that $N_{12}$ is zero for $0\leq R_{12} \leq (1/3)^{3/4}$ which is also the case of
rank-3 border states (refer Eq.~(\ref{Eq:R12N12Rank3Ineq})). It is greater than zero for 
$(1/3)^{3/4}<R_{12}\leq 1$, in which case $N_{12}=(3 R_{12}^{4/3}-1)/2$ or equivalently 
$R_{12}=\left(\left(2 N_{12}+1\right)/3\right)^{3/4}$. This curve is shown in Fig.~\ref{fig:WernerNeg12R12} along with
results from rank-4 matrices derived from a Haar sampling of pure states in $(2,2,4)$ dimensions.

To check that Eq.~(\ref{Eq:R12N12Rank4Ineq}) indeed forms the lower boundary for rank-4 case in 
Fig.~\ref{fig:WernerNeg12R12} the same method from Sec.~\ref{sec:WstateC12method} is employed. Results not presented here then confirm
that the Werner states lie on the boundary of the  $R_{12}$ {\it vs} negativity region and define the extreme curve within which all states lie.
To summarize, for all two-qubit states, including the rank-4 case 
\beq
\label{Eq:R12N12Rank4Ineq10}
N_{12} \leq R_{12} \leq \left(\dfrac{2 N_{12}+1 }{3}\right)^{3/4} \;\; \mbox{for}\;\; N_{12} \geq 0.
\eeq
In Fig.~\ref{fig:AnalyticalCurves2} the boundary curves corresponding to all the ranks of two-qubit density 
matrices are shown. It includes various analytical curves from Eqs.~(\ref{eq:N12R12Curve}), 
(\ref{Eq:N12R12CurveRank3}), (\ref{Eq:R12N12Rank4Ineq}).
The curve corresponding to rank-1 is common to all the ranks. It can be seen that as the rank increases the 
corresponding lower boundaries get lower. 
 {If a two-qubit state is separable then $N_{12}=C_{12}=0$, and in this case the maximum value 
of $R_{12}$ is $(1/3)^{3/4} \approx 0.4387$.}
In other words if for a two-qubit state $R_{12}> (1/3)^{3/4}$, it is necessarily entangled.
{If this criterion is applied to, for example, Werner state it gives clear separable-entangled regions.}

It should be noted that throughout this work possible classes of states with respective boundaries are given and
verified with various numerical and analytical methods
but apart from these there could be other classes of states on the boundaries.

\begin{figure}[t!]
\begin{center}     
\includegraphics[scale=0.35]{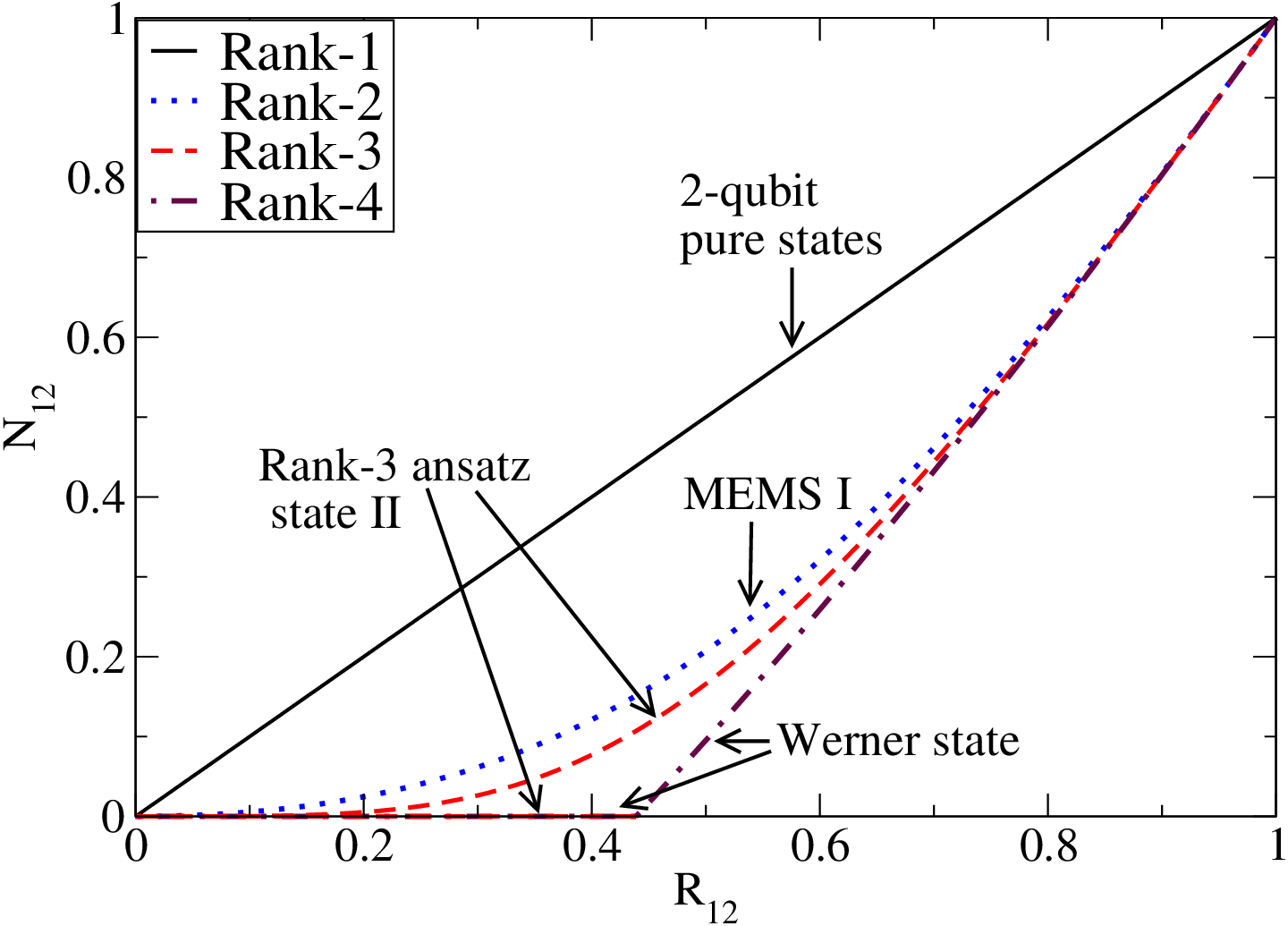}
\caption{(Color online) Boundary curves for $R_{12}$ and the entanglement between two qubits $N_{12}$ having 
density matrix $\rho_{12}$ of all the ranks as per given in Eqs.~(\ref{eq:N12R12Curve}), (\ref{Eq:R12N12Rank3Ineq}) 
and (\ref{Eq:R12N12Rank4Ineq10}). Also shown are various classes of states lying on the respective boundaries.
} 
\label{fig:AnalyticalCurves2} 
\end{center}
\end{figure}

}

\section{Summary and discussions}
\label{sec:Conclusion}

In this paper correlations in bipartite density matrices on ${\cal H}^d \otimes {\cal H}^d$ were studied 
using the quantity $R_{12}~=~d {|\det[\mathcal{R}(\rho_{12}^{T_2})]|}^{1/d^2}$, where 
$\mathcal{R}(\rho_{12}^{T_2})$ is formed by the combined operations of partial transpose and 
realignment on the density matrix. It is based on a simple permutation of the density matrix
and involves no extremization, or even diagonalization. It is proved that $0\leq R_{12} \leq 1$.
Several examples show how they vanish on large classes of separable states including classical-quantum
correlated states, while being maximum ($=1$) on maximally entangled states.
{It is also shown that $R_{12}$ is an entanglement monotone only on two-qubit pure states since
in this case it equals concurrence as well as negativity.
These properties are reminiscent of quantum discord.}
Two qubit density matrices were studied in detail to motivate that this measure captures entanglement in 
the bipartite state along with other multiparty entanglement that maybe present in the purification of such states. 

 In the case of density matrices of rank-2 their purification in terms of  three-qubit states is possible. 
Extensive use of the canonical form of three-qubit pure state is made to make simple connections between 
$R_{12}$, the concurrence, and the tripartite measure of 3-tangle. When the density matrix is of rank-2 
analytical results on the lower and upper bounds of the concurrence in terms of $R_{12}$ are obtained. 
States satisfying the bounds are found to be special, one being the well-known $W$ states, and the other (see comments below)  referred to here as maximally 
3-tangled states ($M3TS$) as they have the property of maximizing the tripartite entanglement for a given 
entanglement {(concurrence)} between two qubits. It is found that this maximization leads to zero entanglement in the other two pairs, reflecting the monogamy of entanglement.

Interestingly, if the entanglement between two pairs of qubits in a tripartite pure state are kept fixed, then the maximization of the 3-tangle does not lead to zero entanglement between the second and the third qubit and leads to a generalization of the $M3TS$. In the case of two-qubit density matrices of ranks three and four, strong evidence is provided that $R_{12}\geq C_{12}$. The physically important subset of $X$ states, which are of rank four in general,
is considered where this is explicitly proved. Upper bounds on $R_{12}$ and the class of states satisfying these bounds are given for all higher ranks as well. As in the case of rank-2 the boundary states are Bell-diagonal (whose purification is the $M3TS$) in the case of rank-3 also it is a Bell-diagonal state, which has been identified and verified numerically. For the rank-4 states the border states are Werner states, which are also special cases of Bell-diagonal states. This the $R_{12}$ {\it vs} concurrence ``phase diagram" is an interesting one which has many special states at the boundaries separating states by rank.

Apart from concurrence another important measure namely the negativity has been compared with $R_{12}$ for 
two-qubit density matrices. Motivations are that negativity, unlike concurrence, can be defined in arbitrary dimensional 
systems, and it is also derived from the operation of partial transpose. Upper and lower bounds on $R_{12}$ in terms of 
negativity and the class of states satisfying these bounds are given for all the ranks. In the case of rank-2 the state is given by MEMS~I, for rank-3 the state is a mixture of two Bell states and a separable pure state orthogonal 
to both of them. While in the case of rank-4 the ansatz state is again the Werner state.
Strong evidences are provided in support of these boundaries and the ansatz states satisfying them.

Since the appearance of the first version of this work, the authors have come to know of various related 
aspects of the central quantity $R_{12}$ and the states $M3TS$. 
In the case of two-qubit case the inequality $R_{12}\geq C_{12}$ is established analytically 
\cite{AntonyMilne2014,AntonyMilne2014BB}. The $M3TS$ states have 
been studied as ``maximal slice" \cite{Carteret00} states and have been shown to violate maximally, for a given 
tangle, a tripartite nonlocality inequality, namely the  Svetlichny inequality 
\cite{Svetlichny1987,RungtaTripartite09}. However the rather simple and natural way in which it appears on 
maximizing the three tangle for a given concurrence (or $R_{12}$) is to our knowledge new.

Thus there is a significance associated with most of the boundary states and hence makes $R_{12}$ an interesting quantity for more detailed investigation. In particular, every two-qubit pure state violates Bell's inequality \cite{bell}, the M3TS violates the Svetlichny inequality, $W$ class of states have zero three-tangle \cite{Vidal00}, Werner states have maximum negativity for given linear entropy, while MEMS~I have maximum concurrence for given linear entropy 
\cite{SatoshiIshizaka2000,Verstraete2001,WJMunro01,TzuChiehWei03,NicholasPeters04}. It will be interesting to 
investigate the significance of the ansatz states I and II. Considering other spectral quantities than the determinant of ${\cal P}(12)$ is possible. More detailed studies and interpretation of such quantities is of interest, and it is hoped that the present work provides sufficient reasons and motivations.

\begin{acknowledgments}
We are very grateful to Antony Milne for bringing our attention to the relationship of $R_{12}$ to ``obesity",
and to Shohini Ghose for pointing out that the $M3TS$ states have previously been studied as the maximal 
slice states. We are also grateful to Karol \.Z{}yczkowski for comments about higher rank boundary states, and 
thank an anonymous referee for suggesting the comparison of negativity and $R_{12}$. UTB is happy to acknowledge
many discussions with M. S. Santhanam, T. S. Mahesh, Sudheer, Swathi, Abhishek, Harshini, Sanku, Deepak, 
Soham, and Anjusha. 
\end{acknowledgments}

\bibliography{reference22013}
\end{document}